\documentclass[12pt]{article}
\usepackage{amsmath}
\usepackage{graphicx,psfrag,epsf}
\usepackage{enumerate}
\usepackage{natbib}
\usepackage{amsmath}
\usepackage{amssymb}
\usepackage{algorithm}
\usepackage{bbm}
\usepackage{fancyhdr}
\usepackage{amsfonts}
\usepackage{subfig,xcolor}
\usepackage{hyperref}
\usepackage{amsthm}
\usepackage{wrapfig}
\usepackage{mathtools}  
\usepackage{enumitem}
\mathtoolsset{showonlyrefs}
\usepackage[noend]{algpseudocode}
\newtheorem{theorem}{Theorem}[section]
\newtheorem{lemma}[theorem]{Lemma}
\newtheorem{definition}[theorem]{Definition}
\usepackage{hyperref}
\usepackage{authblk}

\date{}

\begin{document}

\def\spacingset#1{\renewcommand{\baselinestretch}%
{#1}\small\normalsize} \spacingset{1}


\title{\bf Binary De Bruijn Processes}
  \author{Louise Kimpton$^{1*}$, Peter Challenor$^{1}$, Henry Wynn$^{2}$ \\
  $^{1}$University of Exeter,  $^{2}$London School of Economics \\
  $^{*}$ l.m.kimpton@exeter.ac.uk}
  \maketitle

\bigskip
\begin{abstract}
Binary time series data are very common in many applications, and are typically modelled independently via a Bernoulli process with a single probability of success. However, the probability of a success can be dependent on the outcome successes of past events. Presented here is a novel approach for modelling binary time series data called a binary de Bruijn process which takes into account temporal correlation. The structure is derived from de Bruijn Graphs - a directed graph, where given a set of symbols, $V$, and a ‘word’ length, $m$, the nodes of the graph consist of all possible sequences of $V$ of length $m$. De Bruijn Graphs are equivalent to $m^{\text{th}}$ order Markov chains, where the ‘word’ length controls the number of states that each individual state is dependent on. This increases correlation over a wider area. To quantify how clustered a sequence generated from a de Bruijn process is, the run lengths of letters are observed along with run length properties. Inference is also presented along with two application examples: precipitation data and the Oxford and Cambridge boat race.
\end{abstract}

\noindent%
{\it Keywords: De Bruijn graph, binary, Bernoulli, correlation, run length, Markov chains}  

\def\spacingset#1{\renewcommand{\baselinestretch}%
{#1}\small\normalsize} \spacingset{1}

\section{Introduction}

Although binary time series data are often considered a simple concept, they are typically difficult to model in statistics. Traditionally, binary random variables are modelled independently via a Bernoulli process with a single probability of success ($\pi$) \citep{Feller1950}. However, the assumption that all variables are truly independent is often not sufficient in many applications. The correlation is often dependent on the time between past and future outcomes; given an observed binary time series, it is often the case that the successes in recent trials are more highly correlated than trials further away in time.  For example,  in critical care patients who have contracted COVID-19, daily binary data is collected on whether the patient is connected to a ventilator or not \citep{Douville2021}. The probability of whether the patient is incubated at time $t$ is much more correlated at time $t-k$ when $k$ is small as compared to when $k$ is large.  

The aim of this work is hence to model binary time series data using a correlated process such that there is a temporal varying probability of success; the probability of obtaining a $1$ is higher having observed other $1$'s (and similarly for $0$'s). It is then expected that the successes ($1$'s) will cluster (or anti-cluster in the case of negative correlation). 

Existing common approaches include logistic regression classification, in particular binary logistic regression developed by \cite{LeCessie1994}, and generalised linear models \citep{Diggle1998, Chang2015}.  They allow the marginal probability ($\pi$) of a success ($1$) to vary smoothly in time, and therefore at any time point $x_{t}$ we have a marginal Bernoulli distribution with probability of success equal to $\pi(x_{t})$. However,  only having information regarding the marginal probability of a success at each trial causes any independently sampled sequences to neglect any form of temporal correlation. 

An extension to modelling binary random variables as independent variables is to consider a  Markov property \citep{Billingsley1961, Hillier2006, Isaacson1976}, which captures dependencies between consecutive observations and enables probabilistic inference about future states based on past observations. In a simple binary Markov chain model, there are two states: 0 and 1. The chain transitions between these states based on certain transition probabilities, which dictate the likelihood of moving from one state to another. Similarly, binary time series can be modelled via an auto-regressive framework. The binary AR(1) model was proposed by \cite{McKenzie1985} where a binary time series, $X_{t}$, can be expressed as $X_{t}=(A_{t} - B_{t})X_{t-1} + B_{t}$, where $A_{t}$ and $B_{t}$ are independent binary variables. Although successful in capturing temporal dependencies, there are restrictions on the spread of correlation since the Markov property defines that each outcome is only dependent on the outcome at the previous time step.  Higher-order Markov chains offer a promising extension since the probability of transitioning to a particular state depends on the outcomes of multiple previous states, rather than just the most recent one.

Another related field of literature is correlated multivariate Bernoulli distributions. \cite{Teugels1990} introduces a multivariate Bernoulli distribution, such that given a sequence of $n$ Bernoulli trials, the joint probability of each of the possible $2^{n}$ sequences of binary random variables is specified. The distribution is then parametrised in terms of the central moments so that each of these sequences is further dependent on $2^{n-1}$ parameters. Although the work by \cite{Teugels1990} successfully defines a multivariate Bernoulli distribution, a large number of parameters are required for even short sequences of Bernoulli trials. Secondly, there is no metric to control the spread of correlation across the binary variables.  Both \cite{Fontana2018} and \cite{Euan2020} have extended the work by \cite{Teugels1990} to propose a similar multivariate Bernoulli distribution for a given marginal distribution and correlation matrix, as well as developing a Bernoulli vector autoregressive model.  Other similar approaches are seen in both \cite{Society2017a} and \cite{Society2017}.

Graphical models have also proven to be useful in modelling correlated variables. Similar to the structures used in \cite{Teugels1990},  \cite{Dai2013} use a multivariate Bernoulli distribution to estimate the structure of graphs with binary nodes. The binary nodes are described as random variables, where pairwise correlation is defined in terms of the edges of the graph. Variables are conditionally independent if the associated nodes are not linked by an edge. The authors further extend the model to take account of higher order interactions (not just pairwise interactions) to allow a higher level of structure to be incorporated.  We also point readers to the image processing literature where a similar problem arises when restoring fuzzy images \citep{Besag1986, Besag1974,Abend1965, Bishop2006, Li2016}. This also includes literature relating to cellular automata \citep{Wolfram1959,Agapie2014, Agapie2004}.

In this paper, we propose a framework for modelling binary time series using a correlated multivariate Bernoulli process. We define this process, the de Bruijn process, so named since we incorporate a temporal distance correlation between variables using the structure of de Bruijn graphs \citep{Woude1946, Good1946, Golomb1967, Fredricksen1992}. These are directed graphs, which have nodes consisting of $m$ symbols. We denote the symbols the `letters' of the de Bruijn graph, and the sequence of these letters at each node a de Bruijn `word' of length $m$.  We let the set of letters be $\{0,1\}$ and the de Bruijn structure defines an $m^\text{th}$ order Markov chain \citep{Billingsley1961, Hillier2006, Isaacson1976}. The length $m$ sequences of binary letters then become the states (or words) of the Markov chain where the spread of correlation is controlled by varying $m$.  The de Bruijn graph is a restricted $m^{\text{th}}$ order Markov chain since only 2 edges come in and out of each node. The last $m-1$ letters of the current word must be the same as the first $m-1$ letters of the next word. By defining a distinct probability for each possible transition that can occur in the process, we can control the amount of clustering of each letter.  Our method successfully combines a high level of flexibility in controlling the spread of correlation, but also maintains a high level of simplicity.  We have made use of a graph structure, whilst maintaining a limited number of parameters. Other work involving de Bruijn graphs can be seen in \cite{Hauge1996, Hauge1996a} and \cite{Hunt2002}.

With a de Bruijn process, we are able to model a variety of sequences such that the binary letters can be very clustered together, or alternatively, can be very structured to continuously be alternating (anti-clustered).  In this paper, we define a run length, $R$, to be the number of consecutive $1$'s (or $0$'s) in a row bounded by a 0 (or a 1) at both ends, and use this to calculate a run length distribution. The distribution gives the probability of a run of length $n$ for any $n \in \mathbb{N}^{+}$ in terms of the word length $m$ and the transition probabilities. From this we can then calculate expected run lengths, variance of run lengths and generating functions.  In this paper we apply our run length distribution for a De Bruijn process to precipitation data collected daily from a station in Eskdalemuir, UK. We translate the data to be binary such that $1$ is recorded for when precipitation is present and $0$ otherwise. To remove any seasonal effects, we separate the data into seasons and analyse separately. In each season, we successfully show that the data favours a De Bruijn structure rather than an independent Bernoulli or Markov equivalent.

The remainder of this paper is organised as follows. De Bruijn graphs are explained further in Section \ref{DBG}.  The De Bruijn process is then defined in Section \ref{DBPD} which includes details on the stationary distribution and auto-correlation function. By introducing correlation between binary variables in a sequence, we can control how clustered the $0$'s and $1$'s are. Hence, Section \ref{RLD} gives details of a run length distribution of consecutive numbers of $1$'s in the sequence.  This includes definitions for the expectation, variance and generating functions, as well as examples presented in Section \ref{Examples1}. Section \ref{Inf} then presents methods for inference for both the word length and transition matrix.  Further examples are presented in \ref{Examples2} and applications to precipitation data and the Oxford and Cambridge boat race are given in Section \ref{Expp}. Finally, Section \ref{Conclu} concludes with details of future work.

\section{Preliminaries: De Bruijn Graphs} \label{DBG}

De Bruijn Graphs \citep{Woude1946, Good1946, Golomb1967, Fredricksen1992} are $s$-regular directed graphs consisting of overlapping sequences of symbols. Given a set of $s$ symbols, $V=\{v_{1}, ..., v_{s}\}$, the vertices or nodes of the graph consist of all the possible sequences of $V$. Each graph has $s^{m}$ vertices, where $m$ is the length of each possible sequence given the set of symbols, $V$. The possible nodes are as follows:
\begin{equation}
\begin{split}
V^{m} = \{ (v_{1}&, ..., v_{1},v_{1}), (v_{1}, ..., v_{1}, v_{2}), ..., (v_{1}, ..., v_{1}, v_{s}), (v_{1}, ..., v_{2}, v_{1}), ... ,\\
& (v_{s}, ..., v_{s}, v_{s}) \}.
\end{split}
\end{equation}

Edges in de Bruijn graphs are drawn between node pairs in such a way that the sequences of symbols at connected nodes have overlaps of $m-1$ symbols. An edge is created by removing the first symbol from the node sequence, and adding a new symbol to the end of the sequence from $V$. Thus, from each vertex, $(v_{1}, ..., v_{m}) \in V^{m}$, there is an edge to vertex $(v_{2}, ..., v_{m}, v) \in V^{m}$ for every $v \in V$. There are exactly $s$ directed edges going into each node and $s$ directed edges going out from each node. We denote the symbols, $v$,  the `letters' of the de Bruijn graph, and the sequence of these letters at each node a de Bruijn `word' of length $m$.

By travelling along a path through the de Bruijn graph, chains of letters are generated such that each letter is dependent on the $m$ letters that come before \citep{Fredricksen1992, Ayyer2011, Rhodes2017}. Hence, by altering $m$, we are able to change the dependence structure of the de Bruijn graph. De Bruijn graphs are used in a number of applications including genome sequencing \citep{Tesler2017}, and the mathematics of juggling \citep{Ayyer2015}.

Figure \ref{DBG23} shows two de Bruijn graphs with the set of letters $V=\{v_{1},v_{2}\}=\{0,1\}$, and word lengths $m=2$ (top) and $m=3$ (bottom). The nodes consist of all the possible length two or three sequences of $0$'s and $1$'s where there are two edges coming in and out of each node to give a total of $2m$ nodes and $2m + 1$ edges. 

\begin{figure}[ht]
\centering
\includegraphics[scale=0.8]{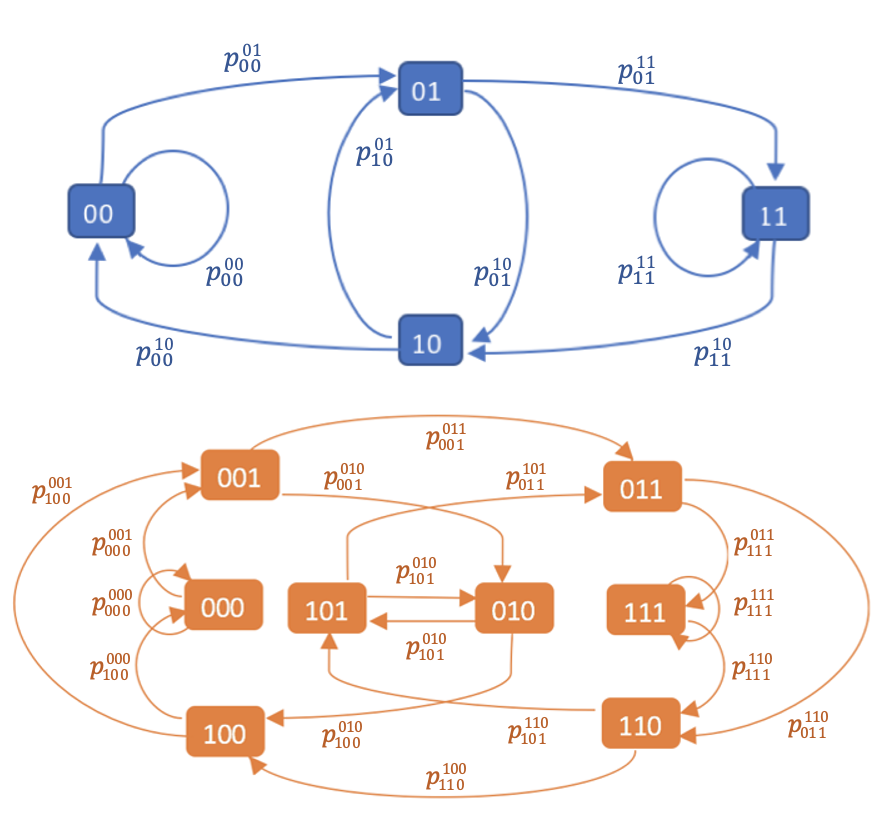}
\caption{Examples of length 2 and 3 de Bruijn graphs with two letters: 0 and 1.}
\label{DBG23}
\end{figure}

\section{Binary De Bruijn Process} \label{DBPD}

Consider a de Bruijn graph with $V = \{0,1\}$, which is a 2-regular directed graph consisting of $2^{m}$ nodes of all possible length $m$ sequences from the letters of the graph, $V$.  The length $m$ sequences at each node are defined as the words of the graph. Let $X = \{X_{1}, X_{2}, \ldots, X_{n}\}$ be a length $n$ sequence of univariate binary random variables such that $X_{i}$ is contained within the set of $s=2$ letters, $V = \{0,1\}$. The sequence $X$ can be written in terms of the words of length $m$ specified by a de Bruijn structure; $X = \{W_{1}, W_{2}, \ldots, W_{n-m+1}\}$, where each $W_{i}$ is an overlapping sequence of $m$ letters. $X$ is defined to be a sequence of order $m$.  We thus formally define a de Bruijn process in Definition \ref{DBPdef}. 
\begin{definition}[de Bruijn Process (DBP)] \label{DBPdef}
\rm{Let $X = \{ W_{1}, W_{2}, \ldots W_{n-m+1} \}$ be a sequence of overlapping random variable de Bruijn words such that $W_{i}$ consist of length $m$ sequences of letters from the set $V^{m}=\{0,1\}^{m}$.  For any positive integer, $t$, and possible states (words), $i_{1}, i_{2}, \ldots, i_{t}$, $X$ is described by a stochastic process in the form of a Markov chain with,}
\begin{equation}
\begin{split}
P \left( W_{t} = i_{t} | W_{t-1} = i_{t-1}, W_{t-2} = i_{t-2}, \ldots, W_{1} = i_{1} \right) &= P \left( W_{t} = i_{t} | W_{t-1} = i_{t-1} \right) \\
&= p_{i_{t-1}}^{i_{t}},
\end{split}
\end{equation}
\rm{where $p_{i_{t-1}}^{i_{t}}$ is the probability of transitioning from the word $i_{t-1}$ to word $i_{t}$.}
\end{definition}

Definition \ref{DBPdef} states $X$ has a Markov property \citep{Billingsley1961, Hillier2006, Isaacson1976} on the word and not the letter, which results in the potential for far more structure to be incorporated into the sequence.  As such, the de Bruijn graph is an order $m$ Markov chain. The length of each word (or order) defines how spread the correlation is over nearby letters (the length over which letters are correlated). For a length $m$ de Bruijn graph, each letter is dependent on the previous $m$ letters, hence the Markov chain becomes a restricted order-$m$ Markov chain of degree exactly two. This structure can be observed in the graph in Figure \ref{DBG23}.  Each word can only transition to other words in which the initial $m-1$ letters are equivalent to the last $m-1$ letters of the previous word. For example, if $m=3$, then the word $001$ can only transition to either the word $010$ or $011$. If m = 1, then the model collapses down to be classically Markov (Markov on the letter), and if $m=0$, this is equivalent to independent Bernoulli trials.

The probability of transitioning from the word $i$ to the word $j$ is given by $p_{i}^{j}$.  Since there is an $m^{\text{th}}$ order Markov structure, each of the transitions form a pair corresponding to whether the next letter is either a $0$ or $1$. For example,if $m=2$, the probability of transition from the word $00$ to the word $01$ is given by $p_{00}^{01}$ where $p_{00}^{00} + p_{00}^{01} = 1$. The transition matrix $T$ of the de Bruijn process is a $2^{m} \times 2^{m}$ matrix, where the $m=2$ case is defined as follows:
\begin{equation}
T = \left(\begin{array}{cccc} 
p_{00}^{00} & p_{00}^{01} & 0 & 0\\
0 & 0 & p_{01}^{10} & p_{01}^{11} \\
p_{10}^{00} & p_{10}^{01} & 0 & 0 \\
0 & 0 & p_{11}^{10} & p_{11}^{11}
\end{array}\right) 
\end{equation}

The binary de Bruijn process is able to model univariate binary time series so that the amount of correlation between variables can be controlled. This is possible by altering the word length and associated transition probabilities, to control the amount of clustering between like letters.  For example, let the probabilities in $T$ be: $\{p_{00}^{01} =0.1, p_{01}^{11} =0.9, p_{10}^{01} =0.1, p_{11}^{11} =0.9\}$ in a length $m=2$ De Bruijn process (top graph in Figure \ref{DBG23}), this will ensure that there is a high level of correlation forcing the letters to cluster together for both $0$'s and $1$'s, and avoiding changes between values. Equivalently, the model can be made very anti-clustered by choosing the transition probabilities that retain the current letter to be small. Examples are given in Section \ref{Examples1} to emphasize the full range of flexibility.

\subsection{Stationary Distribution}

Let $X = \{X_{1}, \ldots, X_{n}\}$ be a length $n$ sequence of univariate binary random variables with outcomes from the set ${\{0,1\}}^{n}$, such that $X$ can be specified in terms of its de Bruijn words of length $m$, $X = \{W_{1}, \ldots, W_{n-m+1}\}$.  The de Bruijn structure for letters implies an $m^{\text{th}}$ order Markov chain, which is equivalent to a $1^{\text{st}}$-order Markov chain on the words. Hence, the stationary distribution (or ergodic distribution) $\pi^{m}$ \citep{Isaacson1976, Jones2001} of a word (or state) $j \in V^{m}$ from an irreducible, persistent, aperiodic Markov chain \citep{Norris1997, Ross2014} with word length $m$ de Bruijn structure is defined as:
\begin{equation}
\pi^{m}(j) = \text{lim}_{t \rightarrow \infty} \hspace{0.2cm} {p_{i}^{j}}^{(t)} ,
\end{equation}
where ${p_{i}^{j}}^{(t)} = P[W_{k+t} = j | W_{k} = i]$ are the transition probabilities at the $t^{\text{th}}$ time step. The stationary probabilities, $\pi^{m}$, must also satisfy the following:
\begin{equation} \label{eqSD}
\begin{split}
\pi^{m}(j) > 0, \hspace{1.5cm}&\hspace{1.5cm} \sum_{j \in V^{m}} \pi^{m}(j) = 1 \hspace{0.5cm} \\
\\
\text{and} \hspace{0.5cm} \pi^{m}(j) &= \sum_{i \in V^{m}} \pi^{m}(i) p_{i}^{j} \\
\Longrightarrow \pi^{m} &= \pi^{m} T
\end{split}
\end{equation} 
for transition matrix $T$,  with ${p_{i}^{j}}^{(t+1)} = \sum_{k} {p_{i}^{k}}^{(n)} p_{k}^{j}$ for $k \in V^{m}$ and by letting $t \rightarrow \infty$.  If ${\pi^{m}}^{(t)}$ gives the marginal probabilities of the de Bruijn words at time $t$, the stationary distribution is hence found when ${\pi^{m}}^{(t)}$ remains unchanged as $t$ increases and gives the marginal or long run time probabilities of the de Bruijn words. Given a long enough sequence $X$ with sufficient burn-in period, $\pi^{m}$ states the proportion of each word occurring in the sequence. 

Consider Equation \eqref{eqSD} for the $m=2$ case where $\pi^{2} = (\pi^{2}(00), \pi^{2}(01), \pi^{2}(10), \pi^{2}(11))$. We want to solve the following:
\begin{equation} 
\left(\begin{array}{cccc} 
\pi^{2}(00) & \pi^{2}(01) & \pi^{2}(10) & \pi^{2}(11)
\end{array}\right) 
=
\left(\begin{array}{c} 
\pi^{2}(00) \\
\pi^{2}(01) \\
\pi^{2}(10) \\
\pi^{2}(11)
\end{array}\right)^{T}
\cdot
\left( 
\begin{array}{cccc} 
p_{00}^{00} & p_{00}^{01} & 0 & 0\\
0 & 0 & p_{01}^{10} & p_{01}^{11} \\
p_{10}^{00} & p_{10}^{01} & 0 & 0 \\
0 & 0 & p_{11}^{10} & p_{11}^{11}
\end{array}\right) 
\end{equation}
After expanding this out, and with some simple rearranging using conservation of probability, we come to the following results:
\begin{equation} \label{2sd1}
\begin{split}
\pi^{2}(00) &= \frac{p_{10}^{00}}{p_{00}^{01}} \pi^{2}(10), \\
\pi^{2}(10) &= \pi^{2}(01), \\
\pi^{2}(11) &= \frac{p_{01}^{11}}{p_{11}^{10}} \pi^{2}(01),
\end{split}
\end{equation}
since $p_{00}^{00} + p_{00}^{01} = 1$ , $p_{10}^{00} + p_{10}^{01} = 1$ and $p_{11}^{10} + p_{11}^{11} = 1$. Therefore if we let $\pi^{2}(01) = \pi^{2}(10) = \alpha$, the system solution becomes $\left(\pi^{2}(00), \pi^{2}(01), \pi^{2}(10), \pi^{2}(11)\right) = \left( \frac{p_{10}^{00}}{p_{00}^{01}} \alpha, \alpha, \alpha, \frac{p_{01}^{11}}{p_{11}^{10}} \alpha \right)$. This shows that the word $01$ has equal probability of occurring in the sequence to the word $10$. This is to be expected since every consecutive sequence of $1$'s (or $0$'s) is bounded by these two words; hence they will occur with equal probability.  The likelihood of the word $00$ occurring is dependent on the transition $p_{10}^{00}$ taking place in conjunction with the transition $p_{00}^{01}$ (to give the sequence $1001$). Hence, the probability of $00$ occurring is the ratio between these two transactions, $\frac{p_{10}^{00}}{p_{00}^{01}}$. The probability of the word $11$ occurring follows by the same reasoning.

Solving the system of equations in \eqref{2sd1} further gives the following stationary distribution for the $m=2$ de Bruijn process:
\begin{equation}
\left(\begin{array}{cccc} 
\pi^{2}(00) & \pi^{2}(01) & \pi^{2}(10) & \pi^{2}(11)
\end{array}\right) 
= \frac{1}{p_{10}^{00} p_{11}^{10} + 2 p_{00}^{01} p_{11}^{10} + p_{00}^{01} p_{01}^{11}}
\left(\begin{array}{c} 
p_{10}^{00} p_{11}^{10} \\
p_{00}^{01} p_{11}^{10} \\
p_{00}^{01} p_{11}^{10} \\
p_{00}^{01} p_{01}^{11}
\end{array}\right)^{T}
\end{equation}

The relationships of the stationary distribution for words of length $m \ge 2$ are given in Theorem \ref{SD1}. To simplify the notation, and to apply to a general word length, the words have been written in terms of the decimal representation of their binary values. This form is repeatedly used for the remainder of this paper and is formally written as, $\sum_{i=1}^{m} k_{i} \hspace{0.1cm} 2^{i-1}$, where $k_{i} \in \{0,1\}$ is each letter in the word; e.g. for the word 010, $\hspace{0.2cm} \sum_{i=1}^{m} k_{i} \hspace{0.1cm} 2^{i-1} = k_{1}2^{0} + k_{2}2^{1} + k_{3}2^{2} = 0 \cdot 1 + 1 \cdot 2 + 0 \cdot 4 = 2$. The stationary distribution for any word length $m$ can be found by solving the system of equations in Theorem \ref{SD1} further using mathematical software such as Maple.

\begin{theorem}[Stationary Distribution Relationships, $m \ge 2$] \label{SD1}
\rm{For a general case $m$, the system of equations for the marginal probabilities, $\pi^{m}$, can be expressed as follows:}
\begin{equation}
\begin{split}
 \pi^{m}(i) &= p_{\frac{1}{2} (i - i \text{ mod } 2)}^{i} \hspace{0.2cm} \pi^{m} \left( \frac{1}{2} (i - i \text{ mod } 2) \right) \\
& \hspace{5cm}  + p_{\frac{1}{2} (i - i \text{ mod } 2) + 2^{m-1)}}^{i} \hspace{0.2cm} \pi^{m} \left( \frac{1}{2} (i - i \text{ mod } 2) + 2^{m-1} \right)  ,
\end{split}
\end{equation}
for $\hspace{0.2cm} i=0:(2^{m}-1)$.
\\
Solving these equations gives the following relationships:
\begin{equation}
\begin{split}
\pi^{m}(0) &= \frac{p_{2^{m-1}}^{0}}{p_{0}^{1}} \hspace{0.2cm} \pi^{m}(2^{m-1}) \\
\pi^{m}(1) &= \pi^{m}(2^{m-1}) \\
\pi^{m}(i) + \pi^{m}(i + 2^{m-1}) &= \pi^{m}(2i) + \pi^{m}(2i + 1) \hspace{1cm} \text{for} \hspace{0.5cm} i=1:(2^{m-1}-2) \\
\pi^{m}(2^{m}-2) &= \pi^{m}(2^{m-1}-1) \\
\pi^{m}(2^{m}-1) &= \frac{p_{2^{m-1}-1}^{2^{m}-1}}{p_{2^{m}-1}^{2^{m}-2}} \hspace{0.2cm} \pi^{m}(2^{m-1}-1)
\end{split}
\end{equation}
\end{theorem}

\begin{proof}
See Appendix
\end{proof}

There are also relationships between the stationary distribution and the transition probabilities which we can express using the law of total probability \citep{Feller1950}.  The law of total probability states that $\pi^{1}(j) = \sum_{n} P(j | i_{n}) \pi^{m}(i_{n})$, so that the probability of a letter ($0$ or $1$) is the weighted average of all the possible words that could generate that letter. In other words, we average over all of the possible starting words to get the probability of obtaining a single $0$ or $1$. 

For example, if $m=2$ we know the following relationships are true:
\begin{equation} \label{mareqn}
\begin{split}
&\pi^{1}(1) = p_{00}^{01}\pi^{2}(00) + p_{01}^{11}\pi^{2}(01) + p_{10}^{01}\pi^{2}(10) + p_{11}^{11}\pi^{2}(11) , \\
&\pi^{1}(0) = p_{00}^{00}\pi^{2}(00) + p_{01}^{10}\pi^{2}(01) + p_{10}^{00}\pi^{2}(10) + p_{11}^{10}\pi^{2}(11) , 
\end{split}
\end{equation}
where $\pi^{1}(0) + \pi^{1}(1) = 1$.  If we further let the transition probabilities be $\{p_{00}^{01}, p_{01}^{11}, p_{10}^{01}, p_{11}^{11}\} = \{0.1,0.9,0.1,0.9\}$), the marginal probabilities for the words are: $\{\pi^{2}(00), \pi^{2}(01), \pi^{2}(10), \pi^{2}(11)\} = \{0.45, 0.05, 0.05, 0.45\}$, and the marginal probabilities for the letters are: $\pi^{1}(0) = 0.5$ and $\pi^{1}(1)=0.5$. There is a $50\%$ chance of getting a $0$ or a $1$, but by including the de Bruijn graph structure into the Markov chains, we are forcing correlation to be accounted for and so the $0$'s and $1$'s will appear in clustered blocks. If we remove the de Bruijn structure and generate a sequence of independent random Bernoulli trials, we would observe the same proportions of each letter, but they would no longer be grouped in blocks.

We can expand Equation \eqref{mareqn} to be applicable for any word length, $m$, as follows:
\begin{equation}
\begin{split}
\pi^{1}(1) &= p_{0 \ldots 00}^{0 \ldots 01} \, \pi^{m}(0 \ldots 0) + p_{0 \ldots 001}^{0 \ldots 011} \, \pi^{m}(0 \ldots 01) + \ldots + p_{1 \ldots 110}^{1 \ldots 101} \, \pi^{m}(1 \ldots 10) \\
& \hspace{1cm} + p_{1 \ldots 1}^{1 \ldots 1} \, \pi^{m}(1 \ldots 1) \\
&= \sum_{i=0}^{2^{m}-1} p_{i}^{(2i + 1) \hspace{0.1cm} \text{mod} 2^{m}} \pi^{m}(i),  
\end{split}
\end{equation}
where $\pi^{1}(0) + \pi^{1}(1) = 1$. 

\subsection{Auto-Correlation Function}

Consider $X = \{X_{1}, \ldots, X_{t}, \ldots, X_{n}\}$ again as a length $n$ sequence of univariate binary random variables. The auto-correlation function defines, on average, the correlation between $X_{t}$ and preceding points. As such, we find the correlation between the given time series and a lagged version of itself over successive time intervals. If we have a strong positive autocorrelation, then there are more chances of $1$'s (or $0$'s) in the next few time steps.

The auto-correlation function is the correlation between $X_{t}$ and $X_{t-k}$ for $k=1, \ldots, t-1$. This is given by:
\begin{equation}
\begin{split}
Corr(X_{t}, X_{t-k}) &= \mathbb{E}\big[(X_{t} - \mathbb{E}[X_{t}])(X_{t-k} - \mathbb{E}[X_{t-k}])\big] \\ &= \mathbb{E}\big[(X_{t} - \pi^{1}(1))(X_{t-k} - \pi^{1}(1))\big]  \\
&= \mathbb{E}\big[ X_{t}X_{t-k} - \pi^{1}(1)X_{t} - \pi^{1}(1)X_{t-k} + \pi^{1}(1)^{2} \big] \\
&= \mathbb{E}\big[ X_{t}X_{t-k} \big] - \pi^{1}(1) \mathbb{E}\big[ X_{t} \big] -\pi^{1}(1) \mathbb{E} \big[ X_{t-k} \big] + \pi^{1}(1)^{2} \\
&= \mathbb{E}\big[ X_{t}X_{t-k} \big] - \pi^{1}(1)^{2},
\end{split}
\end{equation}
where $\mathbb{E}[X_{t}]=\pi^{1}(1)$. To calculate the auto-correlation function, we hence need to find an expression for $\mathbb{E}\big[ X_{t} X_{t-k} \big]$.  Since $X_{t} \in \{0,1\}$, this can be simplified further and is equivalent to finding the joint probability of both $X_{t}$ and $X_{t-k}$ as follows:
\begin{equation}
\begin{split}
\mathbb{E}\big[ X_{t} X_{t-k} \big] &= \sum_{X_{t}=0}^{X_{t}=1} \sum_{X_{t-k}=0}^{X_{t-k}=1} x_{t} x_{t-k} P \big( X_{t} = x_{t}, X_{t-k} = x_{t-k} \big) \\
&= P \big( X_{t}=1, X_{t-k}=1 \big).
\end{split}
\end{equation}

First consider the case for word length $m=2$. To find the joint probability $P \big( X_{t}=1, X_{t-k}=1 \big)$,  for lag $k=1$ we have:
\begin{equation}
\begin{split}
\mathbb{E}\big[ X_{t} X_{t-1} \big] &= P \big( X_{t}=1, X_{t-1}=1 \big)\\
&= \pi^{2}(00) p_{00}^{01} p_{01}^{11} + \pi^{2}(01) p_{01}^{11} p_{11}^{11} + \pi^{2}(10) p_{10}^{01} p_{01}^{11} + \pi^{2}(11) p_{11}^{11} p_{11}^{11}
\end{split}
\end{equation}
For lag $k=2$, this then extends to:
\begin{equation}
\begin{split}
\mathbb{E}\big[ X_{t} X_{t-2} \big] &= P \big( X_{t}=1, X_{t-2}=1 \big) \\
&= P \big( X_{t}=1, X_{t-1}=0, X_{t-2}=1 \big) + P \big( X_{t}=1, X_{t-1}=1, X_{t-2}=1 \big) \\
&= \pi^{2}(00) \big( p_{00}^{01} p_{01}^{10} p_{10}^{01} + p_{00}^{01} p_{01}^{11} p_{11}^{11} \big) + \pi^{2}(01) \big( p_{01}^{11} p_{11}^{10} p_{10}^{01} + p_{01}^{11} p_{11}^{11} p_{11}^{11} \big) \\
& \hspace{1cm} + \pi^{2}(10) \big( p_{10}^{01} p_{01}^{10} p_{10}^{01} + p_{10}^{01} p_{01}^{11} p_{11}^{11} \big) + \pi^{2}(11) \big( p_{11}^{11} p_{11}^{10} p_{10}^{01} + p_{11}^{11} p_{11}^{11} p_{11}^{11} \big).
\end{split}
\end{equation}
Finally, for general lag $k > 2$ we have:
\begin{equation}
\begin{split}
\mathbb{E}\big[ X_{t} X_{t-k} \big] &= P \big( X_{t}=1, X_{t-k}=1 \big) \\
&= \pi^{2}(00) \big( p_{00}^{01} p_{01}^{*} \cdots p_{**}^{*1} + \ldots \big) + \pi^{2}(00) \big( p_{01}^{11} p_{11}^{*} \cdots p_{**}^{*1} + \ldots \big) \\
& \hspace{1cm} \pi^{2}(10) \big( p_{10}^{01} p_{01}^{*} \cdots p_{**}^{*1} + \ldots \big) + \pi^{2}(11) \big( p_{11}^{11} p_{11}^{*} \cdots p_{**}^{*1} + \ldots \big) \\
&= \pi^{2}(00) \cdot \sum_{i=0}^{2^{k-1}-1} p_{00}^{01} \cdot p_{01}^{2+\frac{1}{2^{k-2}}(i - (i \hspace{0.1cm} \text{mod} 2^{k-2}))} \cdot \rho(i,j,k) \\
& \hspace{0.5cm} + \pi^{2}(01) \cdot \sum_{i=0}^{2^{k-1}-1} p_{01}^{11} \cdot p_{11}^{2+\frac{1}{2^{k-2}}(i - (i \hspace{0.1cm} \text{mod} 2^{k-2}))} \cdot \rho(i,j,k) \\
&  \hspace{0.5cm} + \pi^{2}(10) \cdot \sum_{i=0}^{2^{k-1}-1} p_{10}^{01} \cdot p_{01}^{2+\frac{1}{2^{k-2}}(i - (i \hspace{0.1cm} \text{mod} 2^{k-2}))} \cdot \rho(i,j,k) \\
& \hspace{0.5cm} + \pi^{2}(11) \cdot \sum_{i=0}^{2^{k-1}-1} p_{11}^{11} \cdot p_{11}^{2+\frac{1}{2^{k-2}}(i - (i \hspace{0.1cm} \text{mod} 2^{k-2}))} \cdot \rho(i,j,k) \\
&= \sum_{s=0}^{3} \pi^{2}(s) \times \sum_{i=0}^{2^{k-1}-1} p_{s}^{2(s \hspace{0.1cm} \text{mod} 2)+1} \cdot p_{2(s \hspace{0.1cm} \text{mod} 2)+1}^{2+\frac{1}{2^{k-2}}(i - (i \hspace{0.1cm} \text{mod} 2^{k-2}))} \cdot \rho(i,j,k)
\end{split}
\end{equation}
where
\begin{equation}
\rho(i,j,k) = p_{2+\frac{1}{2^{k-2}}(i - (i \hspace{0.1cm} \text{mod} 2^{k-2}))}^{\frac{1}{2^{k-3}}(i - (i \hspace{0.1cm} \text{mod} 2^{k-3}))} \prod_{j=0}^{k-4} p_{[\frac{1}{2^{k-j-3}}(i - (i \hspace{0.1cm} \text{mod} 2^{k-j-3}))] \text{mod} 4}^{[\frac{1}{2^{k-j-4}}(i - (i \hspace{0.1cm} \text{mod} 2^{k-j-4}))] \text{mod} 4} \cdot p_{i \hspace{0.1cm} \text{mod} 4}^{2(i \hspace{0.1cm} \text{mod} 2)+1},
\end{equation}
see Theorem \ref{AC1} for proof extension.  The auto-correlation function for word length $m \ge 2$, and time step lag $k$ is given in Theorm \ref{AC1}.

\begin{theorem}[Auto-Correlation Function, $m \ge 2$] \label{AC1}
\rm{For a de Bruijn process with word length $m$, the auto-correlation function for lag $k=1$ is given as:}
\begin{equation}
\mathbb{E}\big[ X_{t} X_{t-1} \big] = \sum_{s=0}^{2^{m}-1} \pi^{m}(s) \cdot p_{s}^{2(s \hspace{0.1cm} \text{mod} 2^{m-1})+1} p_{2(s \hspace{0.1cm} \text{mod} 2^{m-1})+1}^{4(s \hspace{0.1cm} \text{mod} 2^{m-2})+3}
\end{equation}
This is extended for lag $k>1$ as follows:
\begin{equation}
\begin{split}
\mathbb{E}\big[ X_{t} X_{t-k} \big] &= \sum_{s=0}^{2^{m}-1} \pi^{m}(s) \sum_{i=0}^{2^{k-1}-1} p_{s}^{2(s \hspace{0.1cm} \text{mod} 2^{m-1})+1} \prod_{l=0}^{\text{min}(m,k)-2} p_{2^{l} + \frac{1}{2^{k-l-1}}(i - (i \hspace{0.1cm} \text{mod} 2^{k-l-1})) + 2^{l+1}(s \hspace{0.1cm} \text{mod} 2^{m-l-1})}^{2^{l+1} + \frac{1}{2^{k-l-2}}(i - (i \hspace{0.1cm} \text{mod} 2^{k-l-2})) + 2^{l+2}(s \hspace{0.1cm} \text{mod} 2^{m-l-2})} \\
& \hspace{1cm} \times \rho(i,j,k)
\end{split}
\end{equation}
where
\begin{equation}
\rho(i,j,k) = 
\begin{cases}
p^{2^{k} \hspace{0.1cm} \text{mod} 2^{m} + 2i + 2^{k+1}(s \hspace{0.1cm} \text{mod} 2^{m-k-1})}_{2^{k-1} + i + 2^{k}(s \hspace{0.1cm} \text{mod} 2^{m-k})}  &\quad \text{for } k \le m \\
\\
p_{2^{m-1}+\frac{1}{2^{k-m}}(i - (i \hspace{0.1cm} \text{mod} 2^{k-m}))}^{\frac{1}{2^{k-m-1}}(i - (i \hspace{0.1cm} \text{mod} 2^{k-m-1}))} &\quad \text{for } k > m \\
\hspace{1cm} \times \prod_{j=0}^{k-m-2} p_{[\frac{1}{2^{k-m-j-1}}(i - (i \hspace{0.1cm} \text{mod} 2^{k-m-j-1}))] \text{mod} 2^{m}}^{[\frac{1}{2^{k-m-j-2}}(i - (i \hspace{0.1cm} \text{mod} 2^{k-m-j-2}))] \text{mod} 2^{m}} \cdot p_{i \hspace{0.1cm} \text{mod} 2^{m}}^{2(i \hspace{0.1cm} \text{mod} 2^{m-1})+1}
\end{cases}
\end{equation}
\end{theorem}

\begin{proof}
See Appendix
\end{proof}

\subsection{Joint Bernoulli Distribution} \label{CBD}

Again let $X = \{X_{1}, X_{2}, \ldots, X_{n}\}$ be a length $n$ sequence of univariate binary random variables with outcomes from the set $\{0,1\}^{n}$, which can be modelled using a de Bruijn process with word length $m$ and transition probabilities $p_{i}^{j}$.  The joint distribution $\pi^{n}$ of sequence $X$ is given in Theorem \ref{CBD1}, where $\pi^{m}(i)$ denotes the marginal probability of obtaining the word $m$ as part of the stationary distribution of the process. Due to the de Bruijn framework, this expression is given in terms of the word length, marginal probabilities of the de Bruijn words and the transition probabilities. 

\begin{theorem}[Joint Correlated Bernoulli Distribution $n \ge m$] \label{CBD1}
\rm{For binary random sequence,} $X = \{X_{1}, X_{2}, \ldots, X_{n}\}$\rm{, where} $x= \{x_{1}, x_{2}, \ldots, x_{n}\}$ \rm{is a realisation from} $X$\rm{, the joint probability density is given as follows:}
\begin{equation}
\begin{split}
\pi^{n}(x_{1} & x_{2} \ldots x_{n}) \\
&= P(X_{1} = x_{1}, X_{2} = x_{2}, \ldots, X_{n} = x_{n}) \\
&= \pi^{n}(0...0)^{(1-x_{1})...(1-x_{n})} \times \pi^{n}(0...01)^{(1-x_{1})...(1-x_{n-1})(x_{n})} \times \ldots \\
&\hspace{1cm} \times \pi^{n}(1...10)^{(x_{1})...(x_{n-1})(1-x_{n})} \times \pi^{n}(1...1)^{(x_{1})...(x_{n})} \\
&= \prod_{i=0}^{2^{n}-1} \bigg( \pi^{n}(i) \bigg)^{\prod_{j=1}^{n} \bigg[ \big( x_{j} \big) ^{\big[\frac{1}{2^{n-j}} (i - (i \hspace{0.1cm} \text{mod} 2^{n-j})) \big] \text{mod} 2 } \big( 1 - x_{j} \big) ^{\big[\frac{1}{2^{n-j}} ((2^{n}-i-1) - ((2^{n}-i-1) \hspace{0.1cm} \text{mod} 2^{n-j})) \big] \text{mod} 2 } \bigg]}
\end{split}
\end{equation}
\rm{where}
\begin{equation}
\begin{split}
\pi^{n}(i) &= \sum_{j=0}^{2^{m}-1} \prod_{k=0}^{m-1} \pi^{m}(j) \hspace{0.2cm} p^{2^{k+1} \big( j \hspace{0.1cm}  \text{mod} 2^{m-k-1}\big) + \big[ \frac{1}{2^{n-k-1}} (i - (i \hspace{0.1cm} \text{mod} 2^{n-k-1})) \big] \text{mod} 2^{m} } _{2^{k} \big( j \hspace{0.1cm}  \text{mod} 2^{m-k} \big) + \big[ \frac{1}{2^{n-k}} (i - (i \hspace{0.1cm} \text{mod} 2^{n-k})) \big] \text{mod} 2^{m} } \\
& \hspace{1cm} \times \prod_{s=m}^{n-1} p_{\big[ \frac{1}{2^{n-s}} (i - (i \hspace{0.1cm} \text{mod} 2^{n-s})) \big] \text{mod} 2^{m}} ^{\big[ \frac{1}{2^{n-s-1}} (i - (i \hspace{0.1cm} \text{mod} 2^{n-s-1})) \big] \text{mod} 2^{m} }
\end{split}
\end{equation}
\end{theorem}

\begin{proof}
See Appendix
\end{proof}

Consider a simple example where $x=\{x_{1}, x_{2}, x_{3} \}$ is a realisation from $X = \{X_{1}, X_{2}, X_{3} \}$ with $n=3$ and $m=2$.  By regarding all possible sequences of letters of length $n=3$, the distribution can be expressed as a product of the marginal probabilities of these sequences as follows:
\begin{equation}
\begin{split}
\pi^{3}(x_{1},x_{2},x_{3}) &= \pi^{3}(000)^{(1-x_{1})(1-x_{2})(1-x_{3})} \times \pi^{3}(001)^{(1-x_{1})(1-x_{2})(x_{3})} \times \pi^{3}(010)^{(1-x_{1})(x_{2})(1-x_{3})} \\
& \hspace{0.5cm} \times \pi^{3}(011)^{(1-x_{1})(x_{2})(x_{3})} \times \pi^{3}(100)^{(x_{1})(1-x_{2})(1-x_{3})} \times \pi^{3}(101)^{(x_{1})(1-x_{2})(x_{3})} \\
& \hspace{0.5cm} \times \pi^{3}(110)^{(x_{1})(x_{2})(1-x_{3})} \times \pi^{3}(111)^{(x_{1})(x_{2})(x_{3})} ,
\end{split}
\end{equation}
for any $x_{i} \in \{0, 1\}, $ $ i=1,2,3$. To express each $\pi^{3}(x_{1} x_{2} x_{3})$ in terms of the length $m$ de Bruijn process,  any initial boundary conditions must be included.  All transitions start with an existing word, this includes taking account of all possible starting words. The law of total probability is used as follows:
\begin{equation}
\begin{split}
\pi^{3}(x_{1}x_{2}x_{3}) &= \sum_{j=0}^{3} P(x_{1}x_{2}x_{3} | j) \pi^{2}(j) \\
&= P(x_{1}x_{2}x_{3} | 00) \pi^{2}(00) + P(x_{1}x_{2}x_{3} | 01) \pi^{2}(01) \\
& \hspace{1cm} + P(x_{1}x_{2}x_{3} | 10) \pi^{2}(10) + P(x_{1}x_{2}x_{3} | 11) \pi^{2}(11),
\end{split}
\end{equation}
This can further be expressed in terms of the transition probabilities. For example, letting $x_{1}=1, x_{2}=0, x_{3}=1$ produces the following:
\begin{equation}
\pi^{3}(101) = \pi^{2}(00) p_{00}^{01} p_{01}^{10} p_{10}^{01} + \pi^{2}(01) p_{01}^{11} p_{11}^{10} p_{10}^{01} + \pi^{2}(10) p_{10}^{01} p_{01}^{10} p_{10}^{01} + \pi^{2}(11) p_{11}^{11} p_{11}^{10} p_{10}^{01}.
\end{equation}

\section{Run Length Distribution} \label{RLD}

Let $R$ be a random variable representing the run length of the number of consecutive $1$'s (or $0$'s) in a row bounded by a $0$ (or $1$) at both ends from a sequence $X$. To quantify how clustered a sequence generated from a binary de Bruijn process is, we can consider the run lengths of letters.  Without loss of generality, we will only quantify run lengths of $1$'s, but all of the following results can be extended to be in terms of run lengths of $0$'s.  We also note that all results are conditional on a run existing (i.e. at least a run of $R=1$).

We begin by considering the probability of a run of length $R=r$ for a word length $m=2$ de Bruijn process.  The stationary distributions of the words are given by $\{\pi^{2}(00),\pi^{2}(00),\pi^{2}(00),\pi^{2}(00)\}$ and the transition probabilities are given by $\{ p_{01}^{10}, p_{01}^{11}, p_{11}^{10}, p_{11}^{11} \}$.  The probability of a single $1$ is the transition probability $p_{01}^{10}$, since this gives the probability of transitioning from the word $01$ to the word $10$ (giving the sequence $010$). For higher run lengths we again start with the word $01$, but now transition to the word $11$, with probability $p_{01}^{11}$. For a run of length $R=2$, we finish with the probability $p_{11}^{10}$, whilst for a run length of $R \ge 3$, we would further transition to the same word $11$ with transition probability $p_{11}^{11}$ until a run length of $R=r$ is reached.  Every run must finish with the probability $p_{11}^{10}$ since a zero signifies the end of the run. 

These expressions can be summarised in the form of a recurrence relationship. Once we have obtained the recurrence relationship, we can then use it to define a generating function for the run length distribution. The probability generating function for the run length distribution takes the form $G(y) = \mathbb{E}[y^{r}] = \sum_{r=0}^{\infty} y^{r} P(R = r)$ and has a similar structure to the the geometric distribution.  Once we have reached the word $11$, and continue drawing $1$'s to the chain, then this becomes equivalent to the geometric distribution. 

The recurrence relationship for the run length distribution with word length $m=2$ has the following form:
\begin{equation}
\begin{split}
a_{0} &= 0 \\
a_{1} &= p_{01}^{10} \\
a_{2} &= p_{01}^{11}p_{11}^{10} \\
a_{r+1} &= p_{11}^{11} a_{r} \hspace{1cm} \text{for} \hspace{0.3cm} r \ge 2
\end{split}
\end{equation}
This is solved to find the generating function $G(y) = \sum_{r = 0}^{\infty} a_{r}y^{r}$. Then, multiplying by $y^{r}$ and summing over $r$ gives the following:
\begin{equation}
\begin{split}
\sum_{r\ge 2} a_{r+1}y^{r} &= \sum_{r \ge 2} p_{11}^{11} a_{r} y^{r} \\
\left( a_{3}y^{2} + a_{4}y^{3} + a_{5}y^{4} + \ldots \right) &= p_{11}^{11} \left( \sum_{r \ge 0} a_{r} y^{r} - a_{0} - a_{1}y \right) \\
\frac{1}{y} \left[ \left( a_{0} + a_{1}y + a_{2}y^{2} + \ldots \right) - a_{0} -a_{1}y - a_{2}y^{2} \right] &= p_{11}^{11} \left( G(y) - p_{01}^{10}y \right) \\
\frac{G(y) - p_{01}^{10}y - p_{01}^{11}p_{11}^{10}y^{2} }{y} &= p_{11}^{11} \left( G(y) - p_{01}^{10}y \right) \\
G(y) &= \frac{\left( p_{01}^{11}p_{11}^{10} - p_{01}^{10}p_{11}^{11} \right)y^{2} + p_{01}^{10}y }{1 - p_{11}^{11}y}.
\end{split}
\end{equation}

The power series from the generating function can now be expanded and differentiated to reveal the individual probabilities for each run length. Given this,
the full distribution of run lengths $R$ for a word length $m=2$ de Bruijn process with transition probabilities $\{ p_{01}^{10}, p_{01}^{11}, p_{11}^{10}, p_{11}^{11} \}$ is given in Lemma \ref{RLD2}. This distribution gives the probability of a run of $1$'s of length $R=r$, for any $r \in \mathbb{N}^{+}$.  Specifically for $r \ge 2$, the distribution is equivalent to a modified geometric distribution with probability of success $p_{11}^{10}$, and boundary condition $p_{01}^{11}$ \citep{Johnson2005, Feller1950}.  Although the geometric distribution considers independent trials, we still retain correlation through the Markov property on the words, which ensures the transitions have a unique ordering to match the words to the sequence.

\begin{lemma}[Run Length Distribution, $m=2$] \label{RLD2}
\rm{Let the random variable $R$ denote the run length of a sequence of $1$'s from a sequence $X$. The probability density function of $R$ in terms of length $m=2$ de Bruijn words is given as:}
\begin{equation}
P(R = r) =
\begin{cases}
p_{01}^{10} &\quad \rm{for} \hspace{0.2cm} r=1 \\
p_{01}^{11}(p_{11}^{11})^{r-2}p_{11}^{10} &\quad \rm{for} \hspace{0.2cm} r \ge 2 ,
\end{cases}
\end{equation}
\end{lemma}

\begin{proof}
See Appendix
\end{proof}

We now extend the run length distribution to be applicable for any de Bruijn word length $m \ge 3$.  The general form of the distribution is the same as that seen in Lemma \ref{RLD2}, such that we observe a geometric distribution with probability of success $\left( p_{2^{m}-1}^{2^{m}-2} \right)$, however the boundary condition (or burn-in period) is now much more complex.  We refer the reader to Figure \ref{RLimage}. Each line of this figure shows how a run of length $r$ $1$'s evolves through subsequent transitions. To start a sequence of $1$'s off, the first word must take the form $* \hspace{0.1cm} 01$, where $*$ represents any possible sequence of length $m-2$.  Since there are $2^{m-2}$ possibilities for the letters represented by $*$, we must average over all possibilities using the law of total probability \citep{Feller1950} to get the full run length distribution as follows:
\begin{equation} \label{LTP1}
P(R = r) = \sum_{i=0}^{2^{m-2}-1} P(R = r | *_{i}) \pi^{m-2}(*_{i}),
\end{equation}
where $\pi^{m-2}(*_{i})$ is the marginal probability for the $i$th initial sequence of letters.  Since, $*$ is of length $m-2$, we must then represent each $\pi^{m-2}(*_{i})$ in terms of the length $m$ de Bruijn words.We apply the law of total probability again, such that: 
\begin{equation}
\pi^{m-2}(*_{i}) = \sum_{j=0}^{2^{m-2}-1} P(*_{i} | j) \pi^{m}(j).
\end{equation}

\begin{figure}[ht]
\centering
\includegraphics[scale=0.4]{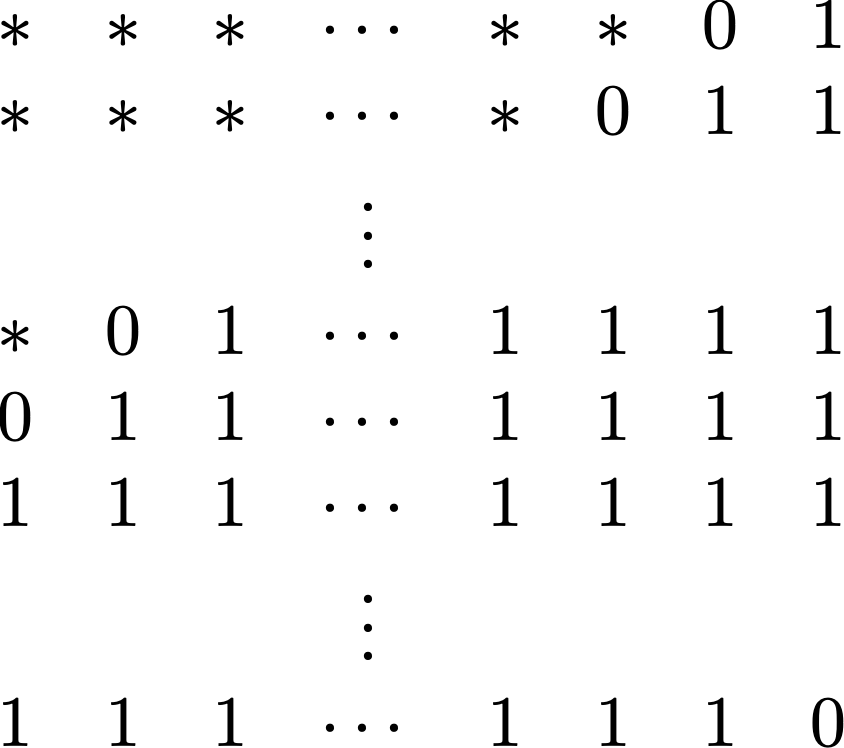}
\caption{Diagram representing the burn-in period for sequences of $0$'s and $1$'s which generate a run of length $r$ $1$'s. Letters represented by $*$ can take either a $0$ or a $1$ for the run to be valid.}
\label{RLimage}
\end{figure}

For longer word lengths $m$, there are not only more possible starting sequences (given by $*$), but the burn-in boundary periods are also longer (as shown in Figure \ref{RLimage}).  If the run does not reach the state containing all $1$'s before a $0$ is drawn, we finish by transitioning to a word of the form $*10$. Alternatively, all runs of length $R=r \ge m$ will end up reaching the all $1$ state with probability $p_{2^{m}-1}^{2^{m}-1}$. The run then ends with the first $0$, resulting in a word of the form $11..10$ with transition probability $p_{2^{m}-1}^{2^{m}-2}$. 

We begin to find a general form for the run length distribution by first finding the recurrence relationship followed by the probability generating function for the run length distribution when $m \ge 3$. The recurrence relationship has the following form:
\begin{equation}
\begin{split}
a_{0} &= 0 \\
a_{1} &= \sum_{i=0}^{2^{m-2}-1} \pi^{m-2}(i) p_{*01}^{*010} \hspace{1.5cm} = \sum_{i=0}^{2^{m-2}-1} \pi^{m-2}(i) p_{4i+1}^{2^{3}(i \hspace{0.1cm} \text{mod } 2^{m-3}) + 2} \\
a_{2} &= \sum_{i=0}^{2^{m-2}-1} \pi^{m-2}(i) p_{*01}^{*011} p_{*011} ^{*0110} \hspace{0.55cm} =   \sum_{i=0}^{2^{m-2}-1} \pi^{m-2}(i) p_{4i+1}^{2^{3}(i \hspace{0.1cm} \text{mod } 2^{m-3})+3} p_{2^{3}(i \hspace{0.1cm} \text{mod } 2^{m-3}) + (2^{2}-1)} ^{2^{4} (i \hspace{0.1cm} \text{mod } 2^{m-4}) + (2^{3}-1) - 1 } \\
\vdots \\
a_{m} &=  \sum_{i=0}^{2^{m-2}-1} \pi^{m-2}(i) p_{*01} ^{*011 } \cdots p_{*1 \ldots 1}^{1 \ldots 10} = \sum_{i=0}^{2^{m-2}-1} \pi^{m-2}(i) \left[ \prod_{j=1}^{m-1} p_{2^{j+1}(i \hspace{0.1cm} \text{mod } 2^{m-j-1}) + (2^{j}-1)} ^{2^{j+2} (i \hspace{0.1cm} \text{mod } 2^{m-j-2}) + (2^{j+1}-1) } \right] p_{2^{m}-1}^{2^{m}-2} \\
\\
a_{r+1} &= p_{2^{m}-1}^{2^{m}-1} a_{r} \hspace{1cm} \text{for} \hspace{0.3cm} r \ge m
\end{split}
\end{equation}
This is solved to find the generating function $G(y) = \sum_{r \ge 0} a_{r}y^{r}$. Then, multiplying by $y^{r}$ and summing over $r$ gives the following:
\begin{equation}
\begin{split}
\sum_{r\ge m} a_{r+1}y^{r} &= \sum_{r \ge m} p_{2^{m}-1}^{2^{m}-1} a_{r} y^{r} \\
\big( a_{m+1}y^{m} + a_{m+2}y^{m+1} + \ldots \big) &= p_{2^{m}-1}^{2^{m}-1} \bigg( \sum_{r \ge 0} a_{r} y^{r} - a_{0} - a_{1}y - \ldots - a_{m-1}y^{m-1} \bigg) \\
\frac{1}{y} \big[ \left( a_{0} + a_{1}y + \ldots \right) - a_{0} -a_{1}y - \ldots - a_{m}y^{m} \big] \hspace{0.2cm} &= p_{2^{m}-1}^{2^{m}-1} \left( G - \sum_{s=0}^{m-1} a_{s}y^{s} \right) \\
\frac{G - \sum_{t=0}^{m} a_{t}y^{t} }{y} &= p_{2^{m}-1}^{2^{m}-1} \left( G - \sum_{s=0}^{m-1} a_{s}y^{s} \right) \\
G - \sum_{t=0}^{m} a_{t}y^{t} &= y p_{2^{m}-1}^{2^{m}-1} G - y p_{2^{m}-1}^{2^{m}-1} \sum_{s=0}^{m-1} a_{s}y^{s} \\
G \left( y p_{2^{m}-1}^{2^{m}-1} - 1 \right) &= y p_{2^{m}-1}^{2^{m}-1} \sum_{s=0}^{m} a_{s}y^{s} - \sum_{t=0}^{m} a_{t}y^{t} - y p_{2^{m}-1}^{2^{m}-1} a_{m}y^{m} \\
G &= \frac{ \left( y p_{2^{m}-1}^{2^{m}-1} - 1 \right) \left(\sum_{s=0}^{m} a_{s}y^{s} \right) - y p_{2^{m}-1}^{2^{m}-1} a_{m}y^{m} }{p_{2^{m}-1}^{2^{m}-1}y - 1} \\
G &= \sum_{s=0}^{m} a_{s}y^{s} + \frac{p_{2^{m}-1}^{2^{m}-1} a_{m} y^{m+1}}{1 - p_{2^{m}-1}^{2^{m}-1}y } \\
G &= \sum_{s=1}^{m} a_{s}y^{s} + \frac{p_{2^{m}-1}^{2^{m}-1} a_{m} y^{m+1}}{1 - p_{2^{m}-1}^{2^{m}-1}y } 
\end{split}
\end{equation}
\rm{where, }
\begin{equation}
\begin{split}
a_{1} &= \sum_{i=0}^{2^{m-2}-1} \pi^{m-2}(i) p_{4i+1}^{2^{3}(i \hspace{0.1cm} \text{mod } 2^{m-3}) + 2} \\
a_{2} &= \sum_{i=0}^{2^{m-2}-1} \pi^{m-2}(i) p_{4i+1}^{2^{3}(i \hspace{0.1cm} \text{mod } 2^{m-3})+3} p_{2^{3}(i \hspace{0.1cm} \text{mod } 2^{m-3}) + (2^{2}-1)} ^{2^{4} (i \hspace{0.1cm} \text{mod } 2^{m-4}) + (2^{3}-1) - 1 } \\
&\vdots \hspace{3cm} \\
a_{m-1} &= \sum_{i=0}^{2^{m-2}-1} \pi^{m-2}(i) \left[ \prod_{j=1}^{m-1} p_{2^{j+1}(i \hspace{0.1cm} \text{mod } 2^{m-j-1}) + (2^{j}-1)} ^{2^{j+2} (i \hspace{0.1cm} \text{mod } 2^{m-j-2}) + (2^{j+1}-1) - \textbf{1}_{j=r} } \right] \\
a_{m} &=  \sum_{i=0}^{2^{m-2}-1} \pi^{m-2}(i) \left[ \prod_{j=1}^{m-1} p_{2^{j+1}(i \hspace{0.1cm} \text{mod } 2^{m-j-1}) + (2^{j}-1)} ^{2^{j+2} (i \hspace{0.1cm} \text{mod } 2^{m-j-2}) + (2^{j+1}-1) } \right] p_{2^{m}-1}^{2^{m}-2}
\end{split}
\end{equation}

Finally, the run length distribution for any run of length $r$ is given in Theorem \ref{RLD3}. We note that the majority of this expression defines the complex boundary conditions of the run.  This is defined as $B(i)$ in the distribution. If we wish to consider run lengths in steady state (ergodic), then the distribution collapses down to $P(R=r) = \Big( p_{2^{m}-1}^{2^{m}-1} \Big)^{r-m} p_{2^{m}-1}^{2^{m}-2}$. This is equivalent to what is shown in the recurrence relationship.

\begin{theorem}[Run Length Distribution, $m \ge 3$] \label{RLD3}
\rm{Let the random variable $R$ denote the run length of a sequence of $1$'s from a sequence $X$. The probability density function of $R$ in terms of length $m$ de Bruijn words is given as:}
\begin{equation}
\begin{split}
P(R = r) & = \begin{cases}
\sum_{i=0}^{2^{m-2}-1} \pi^{m-2}(i) \hspace{0.1cm} p_{*01}^{*010} &\quad \rm{for} \hspace{0.2cm} r=1\\
\sum_{i=0}^{2^{m-2}-1} \pi^{m-2}(i) \hspace{0.1cm}  p_{*01} ^{*011 } \cdots p_{*1 \ldots 1}^{1 \ldots 10}  &\quad \rm{for} \hspace{0.2cm} r=2:m-1\\
\sum_{i=0}^{2^{m-2}-1} \pi^{m-2}(i) \hspace{0.1cm} p_{*01} ^{*011 } \cdots p_{*1 \ldots 1}^{1 \ldots 1}   \left[ \left( p_{1 \ldots 1}^{1 \ldots 1} \right)^{n-m} p_{1 \ldots 1}^{1 \ldots 10} \right] &\quad \rm{for} \hspace{0.2cm} r \ge m 
\end{cases} \\
\\
& = \begin{cases}
\sum_{i=0}^{2^{m-2}-1} \pi^{m-2}(i) \hspace{0.1cm} p_{4i+1}^{2^{3}(i \hspace{0.1cm} \text{mod } 2^{m-3}) + 2} &\quad \rm{for} \hspace{0.2cm} r=1\\
\sum_{i=0}^{2^{m-2}-1} \pi^{m-2}(i) \hspace{0.1cm}  B(i) \cdot \frac{p_{2^{m}-1}^{2^{m}-2}}{p_{2^{m}-1}^{2^{m}-1}}  &\quad \rm{for} \hspace{0.2cm} r=2:m-1\\
\sum_{i=0}^{2^{m-2}-1} \pi^{m-2}(i) \hspace{0.1cm} B(i) \cdot \left[ \left( p_{2^{m}-1}^{2^{m}-1} \right)^{n-m} p_{2^{m}-1}^{2^{m}-2} \right]  &\quad \rm{for} \hspace{0.2cm} r \ge m .\\
\end{cases}
\end{split}
\end{equation}
\rm{where,}
\begin{equation}
\begin{split}
B(i) &= \prod_{j=1}^{\text{min}(m-1,r)} p_{2^{j+1}(i \hspace{0.1cm} \text{mod } 2^{m-j-1}) + (2^{j}-1)} ^{2^{j+2} (i \hspace{0.1cm} \text{mod } 2^{m-j-2}) + (2^{j+1}-1)  } \\
\pi^{m-2}(i) &= \sum_{j=0}^{2^{m}-1} \prod_{k=0}^{m-3} \left[ p^{2^{k+1} (j \hspace{0.1cm} \text{mod } 2^{m-k-1}) + \sum_{s=1}^{k+1} 2^{k-s+1} [ ( \frac{1}{2^{m-s-2}} ( i - (i \hspace{0.1cm} \text{mod } 2^{m-s-2}))) \text{mod } 2 ]}_{2^{k} (j \hspace{0.1cm} \text{mod } 2^{m-k}) + \sum_{s=1}^{k} 2^{k-s} [ ( \frac{1}{2^{m-s-2}} ( i - (i \hspace{0.1cm} \text{mod } 2^{m-s-2}))) \text{mod } 2 ]} \right] \pi^{m}(j)
\end{split}
\end{equation}
\end{theorem}

\begin{proof}
See Appendix
\end{proof}

\subsection{Quantities of Interest} \label{QoI}

From the run length distribution, we can now calculate expected run length and variance of run length \citep{Feller1950}.  First consider the case for $m=2$ de Bruijn processes from the run length distribution given in Lemma \ref{RLD2}. The expected run length is given in Lemma \ref{ERL2} and the variance of run length is given in Lemma \ref{VRL2}.  These are both derived from the probability generating function defined above.

\begin{lemma}[Expected Run Length, $m=2$] \label{ERL2}
\rm{Given the $m=2$ run length distribution in Lemma \ref{RLD2} for random variable $R$, the expected run length is given as:}
\begin{equation}
\mathbb{E}[R] = 1 + \frac{p_{01}^{11}}{p_{11}^{10}}
\end{equation}
\end{lemma}

\begin{proof}
See Appendix 
\end{proof}

\begin{lemma}[Variance of Run Length, $m=2$] \label{VRL2}
\rm{Given the $m=2$ run length distribution in Lemma \ref{RLD2} for random variable $R$, the variance of the run length is given as:}
\begin{equation}
\text{Var}[R] = \frac{p_{01}^{11}}{{p_{11}^{10}}^2} \Big( p_{01}^{10} + p_{11}^{11} \Big)
\end{equation}
\end{lemma}

\begin{proof}
See Appendix
\end{proof}

The expected run length and variance of run length for a length $m \ge 3$ de Bruijn process are given in Theorem \ref{ERLM} and Theorem \ref{VRLM} respectively.  Again, these are both derived from the probability generating function defined above.

\begin{theorem}[Expected Run Length, $m \ge 3$] \label{ERLM}
\rm{Given the run length distribution in Theorem \ref{RLD3} for random variable $R$, the expected run length is given as:}
\begin{equation}
\mathbb{E}[R] = \sum_{s=1}^{m} s \cdot a_{s} + \frac{p_{2^{m}-1}^{2^{m}-1} \Big(1- m \cdot p_{2^{m}-1}^{2^{m}-2} \Big)}{\big( p_{2^{m}-1}^{2^{m}-2} \big)^{2}} a_{m}
\end{equation}
\rm{where,}
\begin{equation}
\begin{split}
a_{1} &= \sum_{i=0}^{2^{m-2}-1} \pi^{m-2}(i) p_{4i+1}^{2^{3}(i \hspace{0.1cm} \text{mod } 2^{m-3}) + 2} \\
a_{2} &= \sum_{i=0}^{2^{m-2}-1} \pi^{m-2}(i) p_{4i+1}^{2^{3}(i \hspace{0.1cm} \text{mod } 2^{m-3})+3} p_{2^{3}(i \hspace{0.1cm} \text{mod } 2^{m-3}) + (2^{2}-1)} ^{2^{4} (i \hspace{0.1cm} \text{mod } 2^{m-4}) + (2^{3}-1) - 1 } \\
&\vdots \hspace{3cm} \\
a_{m-1} &= \sum_{i=0}^{2^{m-2}-1} \pi^{m-2}(i) \left[ \prod_{j=1}^{m-1} p_{2^{j+1}(i \hspace{0.1cm} \text{mod } 2^{m-j-1}) + (2^{j}-1)} ^{2^{j+2} (i \hspace{0.1cm} \text{mod } 2^{m-j-2}) + (2^{j+1}-1) - \textbf{1}_{j=n} } \right] \\
a_{m} &=  \sum_{i=0}^{2^{m-2}-1} \pi^{m-2}(i) \left[ \prod_{j=1}^{m-1} p_{2^{j+1}(i \hspace{0.1cm} \text{mod } 2^{m-j-1}) + (2^{j}-1)} ^{2^{j+2} (i \hspace{0.1cm} \text{mod } 2^{m-j-2}) + (2^{j+1}-1) } \right] p_{2^{m}-1}^{2^{m}-2}
\end{split}
\end{equation}
\end{theorem}

\begin{proof}
See Appendix 
\end{proof}

\begin{theorem}[Variance of Run Length, $m \ge 3$] \label{VRLM}
\rm{Given the run length distribution in Theorem \ref{RLD3} for random variable $R$, the variance of the run length is given as:}
\begin{equation}
\text{Var}[R] = \mathbb{E}[R(R-1)] + \mathbb{E}[R] - \mathbb{E}[R]^{2}
\end{equation}
\rm{where,}
\begin{equation}
\begin{split}
\mathbb{E}[R(R-1)] &= \sum_{s=1}^{m} s(s-1) a_{s} - \frac{p_{2^{m}-1}^{2^{m}-1} \Big( (m^{2} - m) \big( p_{2^{m}-1}^{2^{m}-1} \big)^{2} + (2m^{2} - 2) p_{2^{m}-1}^{2^{m}-1} + m^{2} + m \Big)}{ \big( p_{2^{m}-1}^{2^{m}-2} \big)^{3}} a_{m} \\
\mathbb{E}[R] &= \sum_{s=1}^{m} s \cdot a_{s} + \frac{p_{2^{m}-1}^{2^{m}-1} \Big(1- m \cdot p_{2^{m}-1}^{2^{m}-2} \Big)}{\big( p_{2^{m}-1}^{2^{m}-2} \big)^{2}} a_{m}
\end{split}
\end{equation}
\rm{and}
\begin{equation}
\begin{split}
a_{1} &= \sum_{i=0}^{2^{m-2}-1} \pi^{m-2}(i) p_{4i+1}^{2^{3}(i \hspace{0.1cm} \text{mod } 2^{m-3}) + 2} \\
a_{2} &= \sum_{i=0}^{2^{m-2}-1} \pi^{m-2}(i) p_{4i+1}^{2^{3}(i \hspace{0.1cm} \text{mod } 2^{m-3})+3} p_{2^{3}(i \hspace{0.1cm} \text{mod } 2^{m-3}) + (2^{2}-1)} ^{2^{4} (i \hspace{0.1cm} \text{mod } 2^{m-4}) + (2^{3}-1) - 1 } \\
&\vdots \hspace{3cm} \\
a_{m-1} &= \sum_{i=0}^{2^{m-2}-1} \pi^{m-2}(i) \left[ \prod_{j=1}^{m-1} p_{2^{j+1}(i \hspace{0.1cm} \text{mod } 2^{m-j-1}) + (2^{j}-1)} ^{2^{j+2} (i \hspace{0.1cm} \text{mod } 2^{m-j-2}) + (2^{j+1}-1) - \textbf{1}_{j=n} } \right] \\
a_{m} &=  \sum_{i=0}^{2^{m-2}-1} \pi^{m-2}(i) \left[ \prod_{j=1}^{m-1} p_{2^{j+1}(i \hspace{0.1cm} \text{mod } 2^{m-j-1}) + (2^{j}-1)} ^{2^{j+2} (i \hspace{0.1cm} \text{mod } 2^{m-j-2}) + (2^{j+1}-1) } \right] p_{2^{m}-1}^{2^{m}-2}
\end{split}
\end{equation}
\end{theorem}

\begin{proof}
See Appendix
\end{proof}

We can also calculate the moment generating function of the run length distribution, which takes the form $M(y) = \mathbb{E}[e^{ry}] = \sum_{r=1}^{\infty} e^{ry} P(R = r)$ for $r = 0,1,2,...$. The moment generating functions when $m=2$ and $m \ge 3$ are given in Lemma \ref{MGF2} and Theorem \ref{MGFM} respectively. As with the probability generating function, we can see similarities to the geometric moment generating function, but with the addition of a polynomial term representing the boundary conditions.

\begin{lemma}[Run Length Moment Generating Function, $m=2$] \label{MGF2}
\rm{Given the $m=2$ run length distribution in Lemma \ref{RLD2} for random variable $R$, the moment generating function of the run length is given as follows:}
\begin{equation}
M(y) = \frac{\left( p_{01}^{11}p_{11}^{10} - p_{01}^{10}p_{11}^{11} \right)e^{2y} + p_{01}^{10}e^{y} }{1 - p_{11}^{11}e^{y} }
\end{equation}
\end{lemma}

\begin{proof}
See Appendix
\end{proof}

\begin{theorem}[Run Length Moment Generating Function, $m \ge 3$] \label{MGFM}
\rm{Given the run length distribution in Theorem \ref{RLD3} for random variable $R$, the moment generating function of the run length is given as follows:}
\begin{equation}
M(y) = \sum_{s=1}^{m} a_{s}e^{sy} + \frac{p_{2^{m}-1}^{2^{m}-1} a_{m} e^{(m+1)y}}{1 - p_{2^{m}-1}^{2^{m}-1}e^{y} } 
\end{equation}
\rm{where,} 
\begin{equation}
\begin{split}
a_{1} &= \sum_{i=0}^{2^{m-2}-1} \pi^{m-2}(i) p_{4i+1}^{2^{3}(i \hspace{0.1cm} \text{mod } 2^{m-3}) + 2} \\
a_{2} &= \sum_{i=0}^{2^{m-2}-1} \pi^{m-2}(i) p_{4i+1}^{2^{3}(i \hspace{0.1cm} \text{mod } 2^{m-3})+3} p_{2^{3}(i \hspace{0.1cm} \text{mod } 2^{m-3}) + (2^{2}-1)} ^{2^{4} (i \hspace{0.1cm} \text{mod } 2^{m-4}) + (2^{3}-1) - 1 } \\
&\vdots \hspace{3cm} \\
a_{m-1} &= \sum_{i=0}^{2^{m-2}-1} \pi^{m-2}(i) \left[ \prod_{j=1}^{m-1} p_{2^{j+1}(i \hspace{0.1cm} \text{mod } 2^{m-j-1}) + (2^{j}-1)} ^{2^{j+2} (i \hspace{0.1cm} \text{mod } 2^{m-j-2}) + (2^{j+1}-1) - \textbf{1}_{j=n} } \right] \\
a_{m} &= \sum_{i=0}^{2^{m-2}-1} \pi^{m-2}(i) \left[ \prod_{j=1}^{m-1} p_{2^{j+1}(i \hspace{0.1cm} \text{mod } 2^{m-j-1}) + (2^{j}-1)} ^{2^{j+2} (i \hspace{0.1cm} \text{mod } 2^{m-j-2}) + (2^{j+1}-1) } \right] p_{2^{m}-1}^{2^{m}-2}
\end{split}
\end{equation}
\end{theorem}

\begin{proof}
See Appendix
\end{proof}

Cumulants \citep{McCullagh1987, Wilf1994} offer an alternative to moment generating functions, where the first three are equal to the central moments. The cumulant generating function is found by taking the log of the moment generating function such that $K(y) = \text{log} \big[ M(y) \big] = \text{log} \big[ \mathbb{E} (e^{ny}) \big]$. The $s^{\text{th}}$ cumulant is found by evaluating $K^{s}(0)$, where we will use the fourth central moment, the kurtosis, to find the variance of the sample variance in examples in Section \ref{Examples1}.

\section{Examples I} \label{Examples1}
We now present several examples to show the effects of the transition probabilities ($p_{i}^{j}$) on the distribution of letters in sequences produced from de Bruijn processes (DBP).  It is assumed that binary sequences are unconditional on the starting letters. We hence simulate for a sufficient amount of time to be in steady state, taking a large enough burn-in to ensure that this does not affect any results. 

Each panel in Figure  \ref{Samp4} gives one of four realisations of running a word length $m = 2$ de Bruijn process for $n = 200$ time steps. The marginal probabilities are kept the same ($\pi^{1}(0) = \pi^{1}(1) = 0.5$) so that we can make a better comparison of the spread of letters. We would expect that like letters cluster in larger blocks when the probabilities for remaining at the same letter are kept close to one.  The transition probabilities ($\{p_{00}^{01},p_{01}^{11},p_{10}^{01},p_{11}^{11}\}$)) for the four examples (from top to bottom) take the values $\{0.9, 0.1, 0.9, 0.1\}$, $\{0.5, 0.5, 0.5, 0.5\}$, $\{0.25, 0.75, 0.25, 0.75\}$ and $\{0.1, 0.9, 0.1, 0.9\}$ respectively. The light blue areas represent simulated $0$'s and the dark blue areas represent simulated $1$'s. 

As the sequences progress from top to bottom, the letters are becoming far more clustered together in blocks. Each realisation has equal numbers of $0$'s and $1$'s ($50\%$ of each), but vary in their distribution across the sequence. The top sequence is designed to be constantly swapping between letters (anti-clustered), the second sequence is equivalent to independent Bernoulli trials, the third sequence is slightly more clustered to like letters, and the final sequence is very clustered into blocks of $0$'s and $1$'s. 

These conclusions are also presented in the auto-correlation plots in Figure \ref{ACFP1}, where the autocorrelation is given for 25 lags for each sequence. The anti-clustered sequence shows to alternate between negative and positive correlation up to around lag 12, when the correlation becomes random noise. The bernoulli sequence is shown to be completely random from lag 1 onwards, and both of the correlated sequences are shown to be highly correlated up until around lag 5. A negative correlation appears at around lag 12 indicating the time when the clustering changes letters. 

\begin{figure}[ht]
\centering
\includegraphics[scale=0.55]{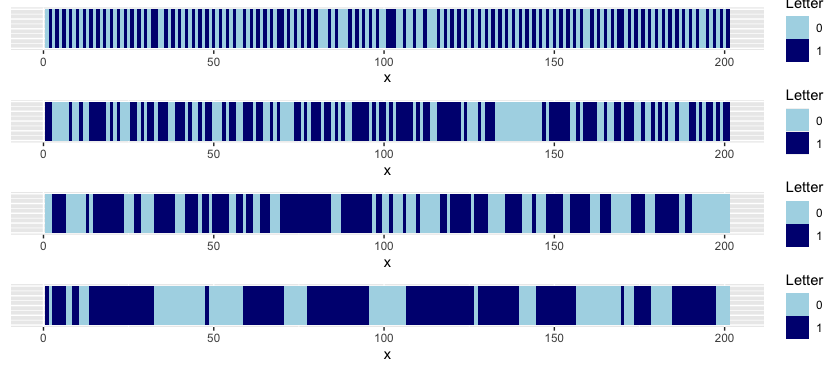}
\caption{Four realisations from de Bruijn processes with letters $0$ and $1$ to show the effects of changing the transition probabilities. From top to bottom, the transition probabilities, $\{ p_{00}^{01}, p_{01}^{11}, p_{10}^{01}, p_{11}^{11} \}$, are: $\{0.9, 0.1, 0.9, 0.1\}$, $\{0.5, 0.5, 0.5, 0.5\}$, $\{0.25, 0.75, 0.25, 0.75\}$ and $\{0.1, 0.9, 0.1, 0.9\}$.}
\label{Samp4}
\end{figure}

\begin{figure}[ht]
\centering
\includegraphics[scale=0.6]{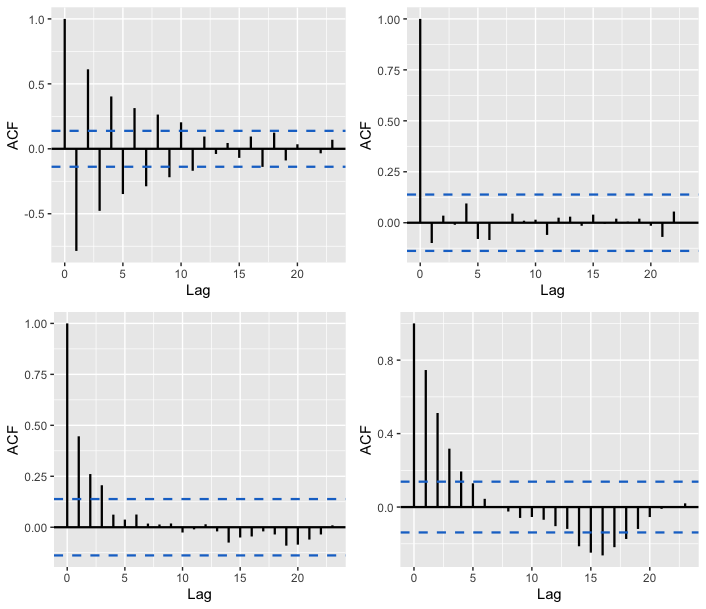}
\caption{Auto-correlation plots for the four examples in Figure \ref{Samp4}. Top left: $\{0.9, 0.1, 0.9, 0.1\}$, top right: $\{0.5, 0.5, 0.5, 0.5\}$, bottom left: $\{0.25, 0.75, 0.25, 0.75\}$, bottom right: $\{0.1, 0.9, 0.1, 0.9\}$.}
\label{ACFP1}
\end{figure}

Figure \ref{Samp4b} shows that it is not necessary to keep the marginal probabilities of the letters equal to 0.5.  In this plot we have two realisations from $m=2$ de Bruijn processes where the top sequence has equivalent marginal probabilities, but the bottom sequence is set to have $\pi^{1}(0)=0.2$ and $\pi^{1}(1)=0.8$. The corresponding transition probabilities for these are $\{0.1,0.8,0.2,0.9\}$ and $\{0.775,0.8,0.8,0.9\}$ respectively. The transition probabilities in this example are defined such that we have $p_{01}^{11}=0.8$ and $p_{11}^{11}=0.9$ for both, but the frequency of $1$'s differs. Hence the lengths of continuous $1$'s in a row is equivalent, but they occur far more often in the latter sequence.

\begin{figure}[ht]
\centering
\includegraphics[scale=0.47]{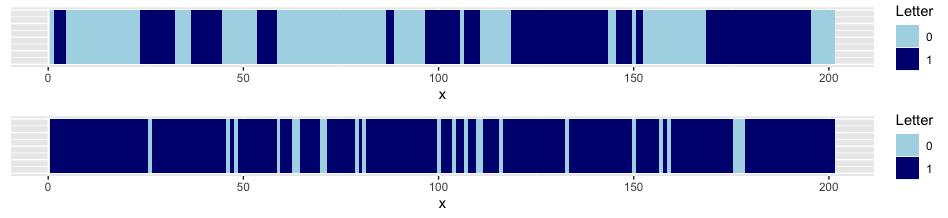}
\caption{Two samples from de Bruijn processes with marginal probabilities $\pi^{(1)}(0)=0.2$, $\pi^{(1)}(1)=0.8$ (top) and $\pi^{(1)}(0)=0.2$, $\pi^{(1)}(1)=0.8$ (bottom). From top to bottom, the transition probabilities, $\{ p_{00}^{01}, p_{01}^{11}, p_{10}^{01}, p_{11}^{11} \}$, are: $\{ 0.1,0.8,0.2,0.9 \}$ and $\{ 0.775,0.8,0.8,0.9 \}$.}
\label{Samp4b}
\end{figure}

The word length of the de Bruijn process also has an effect on the distribution of the letters across the sequence, but it is not as easy to make comparisons due to the different numbers of transition probabilities required for each word. Selecting a longer word length means that more structure can be incorporated through the more intricate transition probabilities, and so it is a mixture of both the word lengths and the transition probabilities that causes the largest impact on the clustering of the sequence.

Next, we will discuss the run lengths of each of the realisations in Figure \ref{Samp4}. We can compare the theoretical results obtained from the de Bruijn run length distribution against the actual run lengths generated from these simulations.  Table \ref{RLProb} shows the probabilities of getting a run of length $r$ for $r = 1, \ldots, 10$ from the run length distribution in Lemma \ref{RLD2}. They are presented in the same order as above so that DBP 1 refers to the de Bruijn process with transition probabilities $\{0.9,0.1,0.9,0.1\}$ and DBP 4 refers to probabilities $\{0.1, 0.9, 0.1, 0.9\}$.  The distributions of run lengths from each sequence are also depicted in the histograms in Figure \ref{RLhist} along with the theoretical run length distributions as seen in Table \ref{RLProb}.

\begin{table}[h!]
\centering
\begin{tabular}{||c | c c c c ||} 
\hline
Run Length, $R$ & DBP 1 & DBP 2 & DBP 3 & DBP 4 \\ [0.5ex] 
\hline\hline
1 & $0.9$ & 0.5 & 0.25 & 0.1 \\ 
2 & $0.09$ & 0.25 & 0.188 & 0.09 \\
3 & $0.009$ & 0.125 & 0.141 & 0.081 \\
4 & $0.0009$ & 0.0625 & 0.105 & 0.0729 \\
5 & $9 \times 10^{-4}$ & 0.0313 & 0.0791 & 0.0656 \\
6 & $9 \times 10^{-5}$ & 0.0156 & 0.0593 & 0.0590 \\  
7 & $9 \times 10^{-6}$ & 0.00781 & 0.0445 & 0.0531 \\
8 & $9 \times 10^{-7}$ & 0.00391 & 0.0334 & 0.0478 \\
9 & $9 \times 10^{-8}$ & 0.00195 & 0.0250 & 0.0430 \\
10 & $9 \times 10^{-9}$ & 0.000977 & 0.0188 & 0.0387 \\ [1ex]
\hline
\end{tabular}
\caption{The probabilities of getting run lengths of $R = 1, ..., 10$ for four different de Bruijn processes of word length $m=2$. The corresponding transition probabilities ($\{ p_{00}^{01}, p_{01}^{11}, p_{10}^{01}, p_{11}^{11} \}$) for these four processes are as follows: DBP 1: $\{0.9, 0.1, 0.9, 0.1\}$, DBP 2: $\{0.5, 0.5, 0.5, 0.5\}$, DBP 3: $\{0.25, 0.75, 0.25, 0.75\}$ , DBP 4: $\{0.1, 0.9, 0.1, 0.9\}$.}
\label{RLProb}
\end{table}

DBP 1 has the highest chance of a short run length,  where the probability for a run length $R \ge 2$ quickly becomes very small. There is a $90\%$ chance of a run length of $R=1$ since this process is designed to constantly be alternating between the two letters. DBP 2 has a $50\%$ chance of a run of length $R=1$. The chance then continues to drop for higher processes until DBP 4 only has a $10\%$ chance of a run of length $R=1$. We also notice that all of the given  probabilities for DBG 3 and DBG 4 stay fairly similar for all run lengths as compared to DBP 1 and DBP 2 which tend to reduce at a fast rate.

The expected run lengths and variance of run lengths for the four examples are given in Table \ref{RLExVar}. Here we have given both the analytical results generated from Lemmas \ref{ERL2} and \ref{VRL2} and the simulated results from the length $n=200$ realisations in Figure \ref{Samp4}. For all cases, both values are shown to agree with on average a difference of $2.18\%$. Hence, this gives confidence that both the simulated and analytic run lengths are calculated correctly.

The expected run length of the structured example (DBP 1) is $1.11$, which is to be expected since the de Bruijn process is designed to be constantly alternating letters. The variance for this process is $0.12$, which confirms that the process is designed to be very structured and does not alter much from having an average run length close to $R=1$. These both match exactly to the theoretical results. The random Bernoulli sequence has both an expected run length of $2$ and variance of $2$. For DBP 3 and DBP 4, since the de Bruijn processes are designed to generate more clustered sequences, the expected run lengths are also shown to increase. This is to be expected since the transition probabilities have a high probability of remaining at the same letter. The variance for these de Bruijn processes also increase significantly; being more likely to generate several very long sequences as well as a few short ones by chance. We observe this effect in the histograms in Figure \ref{RLhist}, where the spread of run lengths gets much larger for DBP 3 and DBP 4. Occasionally we observe a very long run length, but we also occasionally draw very short lengths. There is also likely to be more variance in the simulations, especially if the de Bruijn Markov chains are not run for a sufficient amount of time.

\begin{figure}[ht]
\centering
\includegraphics[scale=0.6]{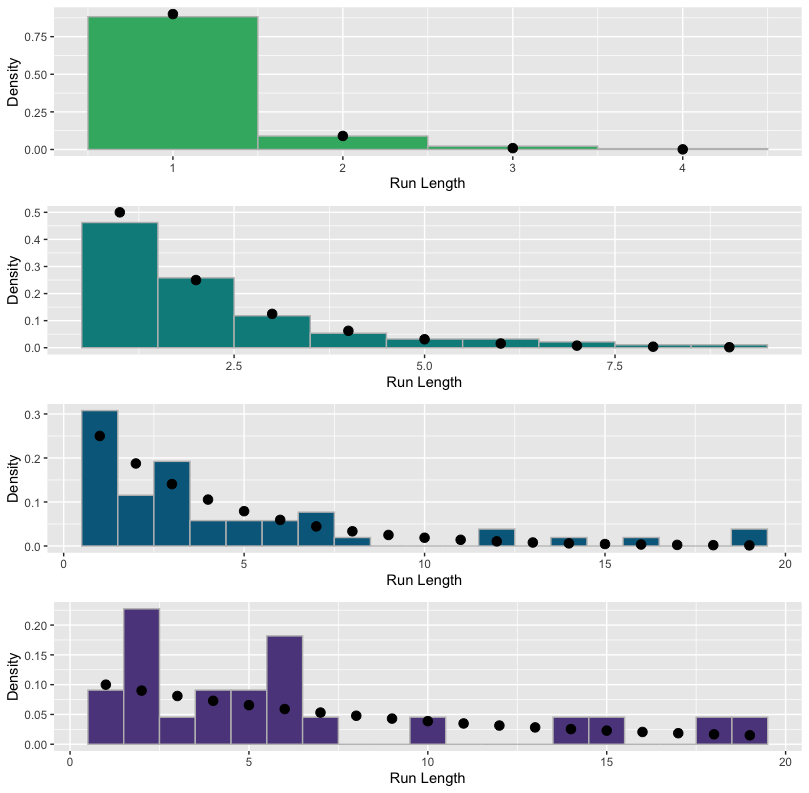}
\caption{Histograms showing the distributions of run lengths of 1's from the de Bruijn process examples in Figure \ref{Samp4}.  Black dots give the theoretical run length distribution derived from Lemma \eqref{RLD2}. The run lengths increase as the de Bruijn process gets increasingly clustered. }
\label{RLhist}
\end{figure}

Table \ref{RLExVar} also shows values for two standard deviations of the sample expected run lengths and two standard deviations of the sample variance of run lengths. These are calculated using the expressions, $\sqrt{\frac{\sigma^{2}}{n}}$ and $\sqrt{\frac{\mu_{4}}{n} - \frac{\sigma^{4} (n-3)}{n(n-1)}}$ respectively, where $n = 200$ and $\mu_{4}$ is the kurtosis. This is the fourth central moment and is found from the fourth differential of the cumulant generating function of the run length distribution evaluated at zero. All the differences between the analytical and simulated expectations and variances are within $\pm$ two standard deviations.

\begin{table}[h!]
\centering
\begin{tabular}{||c | c c c c c c||} 
\hline
D. B. Process & A. Exp & S. Exp & 2 Sd(S. E.) & A. Var & S. Var & 2 Sd(S. V.) \\ [0.5ex] 
\hline\hline
DBP 1 & 1.11 & 1.11 & 0.05 & 0.12 & 0.12 & 0.063  \\ 
DBP 2 & 2 & 2.02 & 0.20 & 2 & 2.21 & 0.66  \\  
DBP 3 & 4 & 4.10 & 0.49 & 12 & 12.07 & 3.83  \\ 
DBP 4 & 10 & 9.47 & 1.34 & 90 & 93.57 & 28.52  \\ [1ex] 
\hline
\end{tabular}
\caption{The theoretical expectation (A. Exp), simulated expectation (S. Exp), two standard deviations of the simulated expectation (Sd(S. E.)), analytical variance (A. Exp), simulated variance (S. Var) and two standard deviations of the simulated variance (Var(S. V.)) of the run length distribution given a realisation of length 200 for four different de Bruijn processes from Figure \ref{Samp4}.}
\label{RLExVar}
\end{table}

\section{Inference} \label{Inf}

Given an observation, $x= \{x_{1}, \ldots, x_{n}\}$, of binary letters from the set $V= \{0,1\}$ with an unknown correlation structure, a de Bruijn process can be fitted to the sequence by estimating both the de Bruijn word length $m$ and the corresponding transition probabilities $p_{i}^{j}$. 

The joint likelihood of the observation $x$ given the transition probabilities $p$ follows from the joint correlated Bernoulli distribution in Theorem \ref{CBD1}.  The likelihood for word length $m=2$ is given in Lemma \ref{TL2} and the likelihood for the general case where $m \ge 1$ is given in Theorem \ref{TLM}.  To simplify notation, $n_{i}^{j}$ denotes the number of times the transition from the word $i$ to the word $j$ takes place in the sequence.  The numerical representation of the binary notation is also included in the general case.  

It is not obvious that the end form of the likelihood gives the joint probability of the sequence. The ordering of the letters in the sequence is fixed and, due to the de Bruijn structure, for each word in the sequence there are only two possible words that can follow. Hence, the distinct number of times each transition occurs in the sequence defines the exact ordering of the letters. 

\begin{lemma}[Transition Likelihood, $m=2$] \label{TL2}
\rm{For a given observation, $x = \{x_{1}, x_{2}, \ldots, x_{n}\}$, the joint likelihood in terms of the length $m=2$ de Bruijn transition probabilities $p=\{p_{00}^{01}, p_{01}^{11}, p_{10}^{01}, p_{11}^{11} \}$ is given as follows:}
\begin{equation}
\begin{split}
\mathcal{L}(p|x) &= (p_{00}^{00})^{\sum_{i=1}^{n-2} (1-x_{i})(1-x_{i+1})(1-x_{i+2})} \times (p_{00}^{01})^{\sum_{i=1}^{n-2} (1-x_{i})(1-x_{i+1})(x_{i+2})} \\
& \hspace{1cm} \times (p_{01}^{10})^{\sum_{i=1}^{n-2} (1-x_{i})(x_{i+1})(1-x_{i+2})} \times (p_{01}^{11})^{\sum_{i=1}^{n-2} (1-x_{i})(x_{i+1})(x_{i+2})} \\
& \hspace{1cm} \times (p_{10}^{00})^{\sum_{i=1}^{n-2} (x_{i})(1-x_{i+1})(1-x_{i+2})} \times (p_{10}^{11})^{\sum_{i=1}^{n-2} (x_{i})(1-x_{i+1})(x_{i+2})} \\
& \hspace{1cm} \times (p_{11}^{10})^{\sum_{i=1}^{n-2} (x_{i})(x_{i+1})(1-x_{i+2})} \times (p_{11}^{11})^{\sum_{i=1}^{n-2} (x_{i})(x_{i+1})(x_{i+2})} \\
&= (p_{00}^{00})^{n_{00}^{00}} \hspace{0.1cm} (p_{00}^{01})^{n_{00}^{01}} \hspace{0.1cm} (p_{01}^{10})^{n_{01}^{10}} \hspace{0.1cm} (p_{01}^{11})^{n_{01}^{11}} \hspace{0.1cm} (p_{10}^{00})^{n_{10}^{00}} \hspace{0.1cm} (p_{10}^{01})^{n_{10}^{01}} \hspace{0.1cm} (p_{11}^{10})^{n_{11}^{10}} \hspace{0.1cm} (p_{11}^{11})^{n_{11}^{11}} \\
&= (1-p_{00}^{01})^{n_{00}^{00}} \hspace{0.1cm} (p_{00}^{01})^{n_{00}^{01}} \hspace{0.1cm} (1-p_{01}^{11})^{n_{01}^{10}} \hspace{0.1cm} (p_{01}^{11})^{n_{01}^{11}} \hspace{0.1cm} (1-p_{10}^{01})^{n_{10}^{00}} \hspace{0.1cm} (p_{10}^{01})^{n_{10}^{01}} \hspace{0.1cm} \\
& \hspace{1cm} \times (1-p_{11}^{11})^{n_{11}^{10}} \hspace{0.1cm} (p_{11}^{11})^{n_{11}^{11}},
\end{split}
\end{equation}
\end{lemma}

\begin{proof}
See Appendix
\end{proof}

\begin{theorem}[Transition Likelihood, $m \ge 1$] \label{TLM}
\rm{For a given observation, $x = \{x_{1}, x_{2}, \ldots, x_{n}\}$, the joint likelihood in terms of the length $m$ de Bruijn transition probabilities $p= \{p_{0}^{1}, \ldots, p_{2^{m}-1}^{2^{m}-1} \}$ is given as follows:}
\begin{equation}
\begin{split}
\mathcal{L}(p|x) &= \left( p_{0 \ldots 0}^{0 \ldots 0} \right)^{\sum_{i=1}^{n-m} (1-x_{i}) \ldots (1-x_{i+m})} \times \left( p_{0 \ldots 00}^{0 \ldots 01} \right)^{\sum_{i=1}^{n-m} (1-x_{i}) \ldots (1-x_{i+m-1})(x_{i+m})} \\
& \hspace{1cm} \times \ldots \times \left( p_{1 \ldots 11}^{1 \ldots 10} \right)^{\sum_{i=1}^{n-m} (x_{i}) \ldots (x_{i+m-1})(1-x_{i+m})} \times \left( p_{1 \ldots 1}^{1 \ldots 1} \right)^{\sum_{i=1}^{n-m} (x_{i}) \ldots (x_{i+m})} \\
&= \left( p_{0 \ldots 0}^{0 \ldots 0} \right)^{n_{0 \ldots 0}^{0 \ldots 0}} \times \left( p_{0 \ldots 00}^{0 \ldots 01} \right)^{n_{0 \ldots 00}^{0 \ldots 01}} \times \ldots \times \left( p_{1 \ldots 11}^{1 \ldots 10} \right)^{n_{1 \ldots 11}^{1 \ldots 10}} \times \left( p_{1 \ldots 1}^{1 \ldots 1} \right)^{n_{1 \ldots 1}^{1 \ldots 1}} \\
&= \prod_{i=0}^{2^{m+1}-1} \left(p_{\frac{1}{2} (i - (i \hspace{0.1cm} \text{mod } 2))}^{i \hspace{0.1cm} \text{mod } 2^{m}} \right) ^{n_{\frac{1}{2} (i - (i \hspace{0.1cm} \text{mod } 2))}^{i \hspace{0.1cm} \text{mod } 2^{m}}}\\
&= \prod_{i=0}^{2^{m}-1} \left( 1 - p_{i}^{(2i+1) \hspace{0.1cm} \text{mod } 2^{m}} \right)^{n_{i}^{((2i+1) \hspace{0.1cm} \text{mod } 2^{m}) - 1}} \left( p_{i}^{(2i+1) \hspace{0.1cm} \text{mod } 2^{m}} \right) ^{n_{i}^{((2i+1) \hspace{0.1cm} \text{mod } 2^{m})}}.
\end{split}
\end{equation}
\end{theorem}

\begin{proof}
See Appendix
\end{proof}

\subsection{Transition Probability Estimation: Maximum Likelihood}

The transition probability likelihood seen in both Lemma \ref{TL2} and Theorem \ref{TLM} can be optimised using a maximum likelihood approach to give estimates of the $2^{m}$ transition probabilities $p_{i}^{j}$.  Due to the different number of transition probabilities required for each word length, it is initially assumed that the word length is known.  

Again,  assume a given observation, $x= \{x_{1}, \ldots, x_{n}\}$ with transition probabilities, $p= \{p_{0}^{1}, \ldots, p_{2^{m}-1}^{2^{m}-1} \}$. Due to the form of the transition likelihood, it is trivial to find a closed form expression for the maximum likelihood estimates of the transition probabilities. Let's first begin with the $m=2$ case. The log-likelihood is given as follows:
\begin{equation}
\begin{split}
\mathcal{L}(p|x) &= (1-p_{00}^{01})^{n_{00}^{00}} \hspace{0.1cm} (p_{00}^{01})^{n_{00}^{01}} \hspace{0.1cm} (1-p_{01}^{11})^{n_{01}^{10}} \hspace{0.1cm} (p_{01}^{11})^{n_{01}^{11}} \hspace{0.1cm} (1-p_{10}^{01})^{n_{10}^{00}} \hspace{0.1cm} (p_{10}^{01})^{n_{10}^{01}} \hspace{0.1cm} \\
& \hspace{1cm} \times (1-p_{11}^{11})^{n_{11}^{10}} \hspace{0.1cm} (p_{11}^{11})^{n_{11}^{11}} \\
\text{log} \big( \mathcal{L}(p|x) \big) &= {n_{00}^{00}} \cdot \text{log} (1-p_{00}^{01}) +  {n_{00}^{01}} \cdot \text{log} (p_{00}^{01})  \\
& \hspace{1cm} + {n_{01}^{10}} \cdot \text{log} (1-p_{01}^{11}) + {n_{01}^{11}} \cdot \text{log} (p_{01}^{11})   \\
& \hspace{1cm} +  {n_{10}^{00}} \cdot \text{log} (1-p_{10}^{01}) + {n_{10}^{01}} \cdot \text{log} (p_{10}^{01}) \\
& \hspace{1cm} + {n_{11}^{10}} \cdot \text{log} (1-p_{11}^{11}) + {n_{11}^{11}} \cdot \text{log} (p_{11}^{11})
\end{split}
\end{equation}
We then proceed by taking partial derivatives of the log-likelihood with respect to the four transition probabilities, $p_{00}^{01}$, $p_{01}^{11}$, $p_{10}^{01}$ and $p_{11}^{11}$.
\begin{equation}
\begin{split}
\frac{\partial \text{log}(\mathcal{L})}{\partial {p_{00}^{01}}} &= \frac{n_{00}^{01}}{p_{00}^{01}} - \frac{n_{00}^{00}}{1 - p_{00}^{01}} \\
\frac{\partial \text{log}(\mathcal{L})}{\partial {p_{01}^{11}}} &= \frac{n_{01}^{11}}{p_{01}^{11}} - \frac{n_{01}^{10}}{1 - p_{01}^{11}} \\
\frac{\partial \text{log}(\mathcal{L})}{\partial {p_{10}^{01}}} &= \frac{n_{10}^{01}}{p_{10}^{01}} - \frac{n_{10}^{00}}{1 - p_{10}^{01}} \\
\frac{\partial \text{log}(\mathcal{L})}{\partial {p_{11}^{11}}} &= \frac{n_{11}^{11}}{p_{11}^{11}} - \frac{n_{11}^{10}}{1 - p_{11}^{11}} 
\end{split}
\end{equation}
Equating each expression to zero and solving for each transition probability gives us the following maximum likelihood estimates:
\begin{equation}
\begin{split}
\hat{p_{00}^{01}} &= \frac{n_{00}^{01}}{n_{00}^{00} + n_{00}^{01}} \\
\hat{p_{01}^{11}} &= \frac{n_{01}^{11}}{n_{01}^{10} + n_{01}^{11}} \\
\hat{p_{10}^{01}} &= \frac{n_{10}^{01}}{n_{10}^{00} + n_{10}^{01}} \\
\hat{p_{11}^{11}} &= \frac{n_{11}^{11}}{n_{11}^{10} + n_{11}^{11}} .
\end{split}
\end{equation}
As we can see, each maximum likelihood estimate of the transition probability $p_{i}^{j}$ is given by the total number of times that the transition from word $i$ to word $j$ takes place in the sequence $x$, divided by the total number of transitions that occur from the word $i$. 

To prove that the likelihoods in both Lemma \ref{TL2} and Theorem \ref{TLM} have a unique maximum likelihood estimate,  we must prove that the distribution takes the form of the exponential family.  If distributions lie in the exponential family as follows,
\begin{equation}
P(X|p) = h(X) \cdot \text{exp} \big( \eta(p) \cdot T(X) - A(p) \big),
\end{equation}
then it can be proven that any stable point of the log-likelihood is a maximum and that there is also at most one of them.  This is equal to the global maximum and thus the log-likelihood function becomes strictly concave.  

Continuing with the $m=2$ example, we first write the distribution as follows:
\begin{equation}
P(X|p) = (1-p_{00}^{01})^{n_{00}^{00}} \hspace{0.1cm} (p_{00}^{01})^{n_{00}^{01}} \hspace{0.1cm} (1-p_{01}^{11})^{n_{01}^{10}} \hspace{0.1cm} (p_{01}^{11})^{n_{01}^{11}} \hspace{0.1cm} (1-p_{10}^{01})^{n_{10}^{00}} \hspace{0.1cm} (p_{10}^{01})^{n_{10}^{01}} \hspace{0.1cm} (1-p_{11}^{11})^{n_{11}^{10}} \hspace{0.1cm} (p_{11}^{11})^{n_{11}^{11}}
\end{equation}
We then take the exponential to give the following:
\begin{equation}
\begin{split}
P(X|p) &= \text{exp} \Big[ n_{00}^{00} \cdot \text{log}(1-p_{00}^{01}) + n_{00}^{01} \cdot \text{log}(p_{00}^{01}) + n_{01}^{10} \cdot \text{log}(1-p_{01}^{11}) + n_{01}^{11} \cdot \text{log}(p_{01}^{11}) \\
& \hspace{1cm} + n_{10}^{00} \cdot \text{log}(1-p_{10}^{01}) + n_{10}^{01} \cdot \text{log}(p_{10}^{01}) + n_{11}^{10} \cdot \text{log}(1-p_{11}^{11}) + n_{11}^{11} \cdot \text{log}(p_{11}^{11}) \Big]
\end{split}
\end{equation}
We finally rearrange and use $N = n_{00}^{00} + n_{00}^{01} + n_{01}^{10} + n_{01}^{11} + n_{10}^{00} + n_{10}^{01} + n_{11}^{10} + n_{11}^{11} = n-1$ for the total number of transitions that take place:
\begin{equation}
\begin{split}
P(X|p) &= \text{exp} \Big[ \Big( n_{00}^{00} \cdot \text{log}(1-p_{00}^{01}) + n_{00}^{01} \cdot \text{log}(p_{00}^{01}) + n_{01}^{10} \cdot \text{log}(1-p_{01}^{11}) + n_{01}^{11} \cdot \text{log}(p_{01}^{11}) \\
& \hspace{1cm} + n_{10}^{00} \cdot \text{log}(1-p_{10}^{01}) + n_{10}^{01} \cdot \text{log}(p_{10}^{01}) + n_{11}^{10} \cdot \text{log}(1-p_{11}^{11}) \Big) \\
& \hspace{2cm} + \big( N - n_{00}^{00} - n_{00}^{01} - n_{01}^{10} - n_{01}^{11} - n_{10}^{00} - n_{10}^{01} - n_{11}^{10} \big) \cdot \text{log}(p_{11}^{11}) \Big] \\
&= \text{exp} \Bigg[ \bigg( n_{00}^{00} \cdot \text{log} \bigg( \frac{1-p_{00}^{01}}{p_{11}^{11}} \bigg) + n_{00}^{01} \cdot \text{log} \bigg( \frac{p_{00}^{01}}{p_{11}^{11}} \bigg) + n_{01}^{10} \cdot \text{log} \bigg( \frac{1-p_{01}^{11}}{p_{11}^{11}} \bigg) + n_{01}^{11} \cdot \text{log} \bigg( \frac{p_{01}^{11}}{p_{11}^{11}} \bigg) \\
& \hspace{1cm} + n_{10}^{00} \cdot \text{log} \bigg( \frac{1-p_{10}^{01}}{p_{11}^{11}} \bigg) + n_{10}^{01} \cdot \text{log} \bigg( \frac{p_{10}^{01}}{p_{11}^{11}} \bigg) + n_{11}^{10} \cdot \text{log} \bigg( \frac{1-p_{11}^{11}}{p_{11}^{11}} \bigg) \bigg) +N \cdot \text{log} \big( p_{11}^{11} \big) \Big] .
\end{split}
\end{equation}
The result hence takes the form of the exponential family with:
\begin{equation}
\begin{split}
 h(X) &= 1, \\ 
\eta(p) &= \Big[ \text{log}(\frac{1-p_{00}^{01}}{p_{11}^{11}}),  \text{log}(\frac{p_{00}^{01}}{p_{11}^{11}}),  \text{log}(\frac{1-p_{01}^{11}}{p_{11}^{11}}),  \text{log}(\frac{p_{01}^{11}}{p_{11}^{11}}),  \text{log}(\frac{1-p_{10}^{01}}{p_{11}^{11}}),  \text{log}(\frac{p_{10}^{01}}{p_{11}^{11}}),  \text{log}(\frac{1-p_{11}^{11}}{p_{11}^{11}}) \Big],  \\ 
T(X) &= \big[ n_{00}^{00}, n_{00}^{01}, n_{01}^{10}, n_{01}^{11}, n_{10}^{00}, n_{10}^{01}, n_{11}^{10}]^{T}, \\
A(p) &= N \cdot \text{log}(p_{11}^{11}).
\end{split}
\end{equation}

The maximum likelihood estimate for a de Bruijn process with general word length $m \ge 1$ is given in Theorem \ref{MLE}. Further proof that this distribution is also part of the exponential family is given in the Appendix.

\begin{theorem}[Maximum Likelihood Estimate, $m \ge 1$] \label{MLE}
\rm{For a given observation $x=\{x_{1},  \ldots, x_{n}\}$,  the maximum likelihood estimate for the transition probabilities $p_{k}^{(2k+1) \text{ mod} 2^{m}}$ of a word length $m$ de Bruijn process is given as follows:
\begin{equation}
\hat{p}_{k}^{(2k+1) \text{ mod} 2^{m}} = \frac{n_{k}^{(2k+1) \text{ mod} 2^{m}}}{n_{k}^{((2k+1) \text{ mod} 2^{m}) -1} + n_{k}^{(2k+1) \text{ mod} 2^{m}}},
\end{equation}
for $k = 0,1, \ldots, m-1$.}
\end{theorem}

\begin{proof}
See Appendix
\end{proof}

Uncertainties on the maximum likelihood estimates can then be calculated using the Fisher information \citep{Feller1950}, $I(p_{i}^{j}) = -\mathbb{E} \left[ \frac{\partial^{2} \text{log}(\mathcal{L})}{\partial {p_{i}^{j}}^{2}} \right]$. Note that $p_{i}^{j}$ still denotes the probability of transitioning from the word $i$ to the word $j$, and $\text{log} (\mathcal{L})$ is the natural log of the likelihood given in Theorem \ref{TLM}. The expectation of a function $g(x_{1}, \ldots, x_{n})$ with respect to the $x_{i}$ is given as:
\begin{equation}
\mathbb{E}[g(X_{1}, \ldots, X_{n})] = \sum_{x_{1}=0}^{x_{1}=1} \cdots \sum_{x_{n}=0}^{x_{n}=1} g(x_{1}, \ldots, x_{n}) \pi(x_{1}, \ldots, x_{n}),
\end{equation}
where $\pi(x_{1}, \ldots, x_{n})$ is taken to be the joint correlated Bernoulli distribution from Theorem \ref{CBD1}. The expectation is hence dependent on the length of the sequence $n$, word length and transition probabilities.  

First consider the case for the transition probability $p_{00}^{00}$ where $m=2$ and $n=3$. The expectation of the double derivative with respect to $p_{00}^{00}$ is as follows:
\begin{equation}
\begin{split}
\mathbb{E}\left[\frac{\partial^{2}\text{log}(\mathcal{L})}{\partial {p_{00}^{00}}^{2}}\right] &= -\frac{1}{\left(p_{00}^{00}\right)^{2}} \mathbb{E}\left[ (1-x_{1})(1-x_{2})(1-x_{3}) \right] \\
&= -\frac{1}{\left(p_{00}^{00}\right)^{2}} \sum_{x_{1}=0}^{x_{1}=1} \sum_{x_{2}=0}^{x_{2}=1} \sum_{x_{3}=0}^{x_{3}=1} (1-x_{1})(1-x_{2})(1-x_{3}) \pi^{(3)}(x_{1}, x_{2}, x_{3}) \\
&= -\frac{1}{\left(p_{00}^{00}\right)^{2}} \pi^{3}(000)
\end{split}
\end{equation}
Expanding this for general $n \ge 3$, results in:
\begin{equation}
\begin{split}
\mathbb{E}\left[\frac{\partial^{2}\text{log}(\mathcal{L})}{\partial {p_{00}^{00}}^{2}}\right] &= -\frac{1}{p_{00}^{00}} \sum_{i=0}^{n-3} \sum_{j=0}^{2^{n-3}-1} \pi^{n} \left( 2^{3} j - \left( 2^{3} - 1 \right) \left( j \hspace{0.1cm} \text{mod} 2^{i} \right) \right) \\
&= -\frac{1}{p_{00}^{00}} \sum_{i=0}^{n-3} \sum_{j=0}^{2^{n-3}-1} \pi^{n} \left( 8 j - 7 \left( j \hspace{0.1cm} \text{mod} 2^{i} \right) \right)
\end{split}
\end{equation}
Then for any transition probability, $p$, the general result for $m=2$ is as follows:
\begin{equation}
\begin{split}
\mathbb{E}\left[\frac{\partial^{2}\text{log}(\mathcal{L})}{\partial {p_{k}}^{2}}\right] &= -\frac{1}{p_{k}} \sum_{i=0}^{n-3} \sum_{j=0}^{2^{n-3}-1} \pi^{n} \left( 2^{3} j + 2^{i}k - \left( 2^{3} - 1 \right) \left( j \hspace{0.1cm} \text{mod} 2^{i} \right) \right) \\
&= -\frac{1}{p_{k}} \sum_{i=0}^{n-3} \sum_{j=0}^{2^{n-3}-1} \pi^{n} \left( 8 j + 2^{i}k - 7 \left( j \hspace{0.1cm} \text{mod} 2^{i} \right) \right),
\end{split}
\end{equation}
for $\{p_{0}, p_{1}, p_{2}, p_{3}, p_{4}, p_{5}, p_{6}, p_{7}\} = \{p_{00}^{00}, p_{00}^{01}, p_{01}^{10}, p_{01}^{11}, p_{10}^{00}, p_{10}^{01},p_{11}^{10}, p_{11}^{11}\}$ and $n \ge 3$. The fisher information is then given by $I(p_{k})=- \mathbb{E} \left[\frac{\partial^{2}\text{log}(\mathcal{L})}{\partial {p_{k}}^{2}}\right]$.

Following from this, the Fisher information for the general case for word length $m$ is given in Theorem \ref{FI}.

\begin{theorem}[Fisher information, $m \ge 1$] \label{FI}
\rm{The Fisher information, } $I(p_{k}) = -\mathbb{E}\left[\frac{\partial^{2}\text{log}(\mathcal{L})}{\partial {p_{k}}^{2}}\right]$, \rm{ for each transition probability $p_{k} = p_{\frac{1}{2}(k - (k \text{ mod} 2))}^{k \text{ mod} 2^{m}}$ for $k=0, 1, \ldots, 2^{m+1}-1$, sequence length $n$ and word length $m$ is given by:}
\begin{equation}
\begin{split}
I(p_{k}) &= -\mathbb{E}\left[\frac{\partial^{2}\text{log}(\mathcal{L})}{\partial {p_{k}}^{2}}\right] \\
&= \frac{1}{p_{k}} \sum_{i=0}^{n-m-1} \sum_{j=0}^{2^{n-m-1}-1} \pi^{n} \left( 2^{m+1} j + 2^{i}k - \left( 2^{m+1} - 1 \right) \left( j \hspace{0.1cm} \text{mod} 2^{i} \right) \right)
\end{split}
\end{equation}
\end{theorem}

\begin{proof}
See Appendix
\end{proof}

\subsection{Transition Probability Estimation: Bayesian Approach}

Alternatively, the likelihood in Theorem \ref{TLM} can be used to estimate the transition probabilities using Bayes' theorem. The advantage to this is that the prior distribution can be used to specify any prior knowledge already known about the transition probabilities. For example, it may be known that the sequence is very clustered in blocks of $1$'s, and this can be incorporated into the prior distribution to put higher weighting onto the transition that remains at the all $1$ word. Due to the form of the likelihood (Theorem \ref{TLM}),  a Beta prior of the form, $P(p) = \frac{\Gamma(\alpha + \beta)}{\Gamma(\alpha)\Gamma(\beta)}p^{\alpha -1}(1-p)^{\beta -1}$, for $\alpha > 0$ and $\beta > 0$,  is used to produce the posterior distribution for the transition probabilities $p$. 

First consider the de Bruijn word length $m=2$ case given observation, $x= \{x_{1}, \ldots, x_{n}\}$ with transition probabilities, $p= \{p_{0}^{1}, \ldots, p_{2^{m}-1}^{2^{m}-1} \}$.The transition likelihood in Lemma \ref{TL2} can be combined with a Beta prior to obtain the posterior distribution as follows:
\begin{equation}
\begin{split}
P(p|X) &= \frac{P(x|p)P(p)}{\int P(x|p)P(p) dp} \\
&\propto \hspace{0.2cm}  P(x|p)P(p) \\
&= \hspace{0.2cm} (1 - p_{00}^{01})^{n_{00}^{00}+\beta_{1}-1} \hspace{0.1cm} (p_{00}^{01})^{n_{00}^{01}+\alpha_{1}-1} \hspace{0.2cm}  \\
& \hspace{0.8cm} \times (1 - p_{01}^{11})^{n_{01}^{10}+\beta_{2}-1} \hspace{0.1cm} (p_{01}^{11})^{n_{01}^{11}+\alpha_{2}-1} \hspace{0.2cm}  \\
& \hspace{0.8cm} \times (1 - p_{10}^{01})^{n_{10}^{00}+\beta_{3}-1} \hspace{0.1cm} (p_{10}^{01})^{n_{10}^{01}+\alpha_{3}-1} \hspace{0.2cm}  \\ 
& \hspace{0.8cm} \times (1 - p_{11}^{11})^{n_{11}^{10}+\beta_{4}-1} \hspace{0.1cm} (p_{11}^{11})^{n_{11}^{11}+\alpha_{4}-1}.
\end{split}
\end{equation}

The posterior distribution for the de Bruijn process transition probabilities is hence a product of beta densities.  Since the prior and posterior distributions are conjugate, the following can be stated:
\begin{equation} \label{MEF2}
\begin{split}
\int P(x|p)P(p) dp = & \frac{\Gamma(n_{00}^{00}+\beta_{1})\Gamma(n_{00}^{01}+\alpha_{1})}{\Gamma(n_{00}^{00}+n_{00}^{01}+\beta_{1}+\alpha_{1})} \times \frac{\Gamma(n_{01}^{10}+\beta_{2})\Gamma(n_{01}^{11}+\alpha_{2})}{\Gamma(n_{01}^{10}+n_{01}^{11}+\beta_{2}+\alpha_{2})} \times \\
& \frac{\Gamma(n_{10}^{00}+\beta_{3})\Gamma(n_{10}^{01}+\alpha_{3})}{\Gamma(n_{10}^{00}+n_{10}^{01}+\beta_{3}+\alpha_{3})} \times 
\frac{\Gamma(n_{11}^{10}+\beta_{4})\Gamma(n_{11}^{11}+\alpha_{4})}{\Gamma(n_{00}^{10}+n_{11}^{11}+\beta_{4}+\alpha_{4})},
\end{split}
\end{equation}
which specifies the model evidence. For the general case when $m \ge 1$, the posterior distribution for the transition probabilities is given in Theorem \ref{MEM}.

\begin{theorem}[Posterior Distribution for de Bruijn Probability Transitions, $m \ge 1$] \label{MEM}
\rm{For a given observation $x=\{x_{1},  \ldots, x_{n}\}$ and transition probabilities, $p= \{p_{0}^{1}, \ldots, p_{2^{m}-1}^{2^{m}-1} \}$, applying Bayes' theorem, the posterior distribution of the de Bruijn transition probabilities is:}
\begin{equation}
\begin{split}
P(p | x) & = {P(x|p,m) P(p|m)}{P(x)} \\
& = \frac{P(x|p,m)P(p|m)}{\int P(x|p,m)P(p|m) dp}
\end{split}
\end{equation}
\rm{where,}
\begin{equation}
\begin{split}
P(x|p,m)P(p|m) &= \prod_{i=0}^{2^{m}-1} (1 - p_{i}^{((2i+1) \hspace{0.1cm} \text{mod } 2^{m})})^{n_{i}^{((2i+1) \hspace{0.1cm} \text{mod } 2^{m})-1 }+\beta_{i+1}-1} \\
& \hspace{2cm} \times (p_{i}^{((2i+1) \hspace{0.1cm} \text{mod } 2^{m})})^{n_{i}^{((2i+1) \hspace{0.1cm} \text{mod } 2^{m}) }+\alpha_{i+1}-1}
\end{split}
\end{equation}
\rm{and}
\begin{equation}
\int P(x|p,m)P(p|m) dp = \prod_{i=0}^{2^{m}-1} \frac{\Gamma(n_{i}^{((2i+1) \hspace{0.1cm} \text{mod } 2^{m}) - 1} + \beta_{i+1})\Gamma(n_{i}^{((2i+1) \hspace{0.1cm} \text{mod } 2^{m})}) + \alpha_{i+1})}{\Gamma(n_{i}^{((2i+1) \hspace{0.1cm} \text{mod } 2^{m}) - 1} + n_{i}^{((2i+1) \hspace{0.1cm} \text{mod } 2^{m})} +\beta_{i+1}+\alpha_{i+1})}
\end{equation}
\end{theorem}

\begin{proof}
See Appendix
\end{proof}

\subsection{Word Length Estimation: Profile Maximum Likelihood}

Estimating the de Bruijn word length $m$ for a binary observation $x=\{x_{1},  \ldots, x_{n}\}$ is a more challenging problem since $m$ must take an integer value and a different number of transition probabilities are required for different word lengths.  One approach is to consider a model selection approach by considering different word lengths as models for comparison. The Akaike Information Criterion (AIC) is a measure often used for model selection. It balances the trade-off between the goodness of fit of the model and the complexity of the model, particularly with the number of parameters. AIC provides a numerical score that quantifies the relative quality of different models, with lower values indicating better-fitting models.

The AIC in terms of the de Bruijn process is calculated based on the likelihood of the process, which we have already specified in Lemma \ref{TL2} and Theorem \ref{TLM}. This is given as follows:
\begin{equation}
\begin{split}
\text{AIC} &= 2 \cdot 2^{m} - 2 \cdot \text{log} \big( \mathcal{L}(X|p) \big) \\
&= 2^{m+1} - 2 \cdot \text{log} \big( \mathcal{L}(X|p) \big),
\end{split}
\end{equation}
since we have $2^{m}$ parameters to be estimated in a word length $m$ de Bruijn process. The main issue with this is that the likelihood is conditional on the transition probabilities, which we do not know prior to estimating the word length. To solve this problem, we instead look to calculating the profile likelihood just in terms of the sequence $x$ to remove the transition probabilities from the calculation.

We define the transition likelihood for a sequence $x$ when $m=2$ as follows:
\begin{equation}
\mathcal{L}(p|x) = (p_{00}^{00})^{n_{00}^{00}} \hspace{0.1cm} (p_{00}^{01})^{n_{00}^{01}} \hspace{0.1cm} (p_{01}^{10})^{n_{01}^{10}} \hspace{0.1cm} (p_{01}^{11})^{n_{01}^{11}} \hspace{0.1cm} (p_{10}^{00})^{n_{10}^{00}} \hspace{0.1cm} (p_{10}^{01})^{n_{10}^{01}} \hspace{0.1cm} (p_{11}^{10})^{n_{11}^{10}} \hspace{0.1cm} (p_{11}^{11})^{n_{11}^{11}}.
\end{equation}
Note that we are not using the total law of probability for Markov chains to simplify the likelihood to four parameters in this case.  Given this expression, we calculate the profile likelihood $\mathcal{L}_{\text{P}}(X|p)$ by integrating over each of the transition probabilities:
\begin{equation}
\begin{split}
\mathcal{L}_{\text{P}}(p|x) &= \int_{0}^{1} \cdots \int_{0}^{1} (p_{00}^{00})^{n_{00}^{00}} \hspace{0.1cm} (p_{00}^{01})^{n_{00}^{01}} \hspace{0.1cm} (p_{01}^{10})^{n_{01}^{10}} \hspace{0.1cm} (p_{01}^{11})^{n_{01}^{11}} \hspace{0.1cm} (p_{10}^{00})^{n_{10}^{00}} \hspace{0.1cm} (p_{10}^{01})^{n_{10}^{01}} \hspace{0.1cm} \\
& \hspace{2.5cm} \times (p_{11}^{10})^{n_{11}^{10}} \hspace{0.1cm} (p_{11}^{11})^{n_{11}^{11}} \hspace{0.5cm} dp_{00}^{00} \, dp_{00}^{01} \, dp_{01}^{10} \, dp_{01}^{11} \, dp_{10}^{00} \,  dp_{10}^{01} \, dp_{11}^{10} \, dp_{11}^{11} \\
&= \Bigg[ \frac{(p_{00}^{00})^{n_{00}^{00}+1}}{n_{00}^{00} + 1}  \cdot \frac{(p_{00}^{01})^{n_{00}^{01}+1}}{n_{00}^{01} + 1} \cdot  \frac{(p_{01}^{10})^{n_{01}^{10}+1}}{n_{01}^{10} + 1} \cdot \frac{(p_{01}^{11})^{n_{01}^{11}+1}}{n_{01}^{11} + 1} \cdot \frac{(p_{10}^{00})^{n_{10}^{00}+1}}{n_{10}^{00} + 1} \cdot  \frac{(p_{10}^{01})^{n_{10}^{01}+1}}{n_{10}^{01} + 1} \\
& \hspace{2.5cm} \times  \frac{(p_{11}^{10})^{n_{11}^{10}+1}}{n_{11}^{10} + 1} \cdot  \frac{(p_{11}^{11})^{n_{11}^{11}+1}}{n_{11}^{11} + 1} \Bigg]_{0}^{1} \\
&= \frac{1}{n_{00}^{00}+1} \cdot \frac{1}{n_{00}^{01}+1} \cdot  \frac{1}{n_{01}^{10}+1} \cdot \frac{1}{n_{01}^{11}+1} \cdot \frac{1}{n_{10}^{00}+1} \cdot \frac{1}{n_{10}^{01}+1} \cdot \frac{1}{n_{11}^{10}+1} \cdot \frac{1}{n_{11}^{11}+1}
\end{split}
\end{equation}
The AIC for the profile likelihood when $m=2$ is then calculated as:
\begin{equation}
AIC_{\text{P}} = 8 - 2 \cdot \text{log} \big( \mathcal{L}_{\text{P}} (p|x) \big)
\end{equation}

The AIC along with the profile likelihood for general word length $m \ge 1$ is given in Theorem \ref{AICC}. In the set up of the de Bruijn process, we make the assumption that the word length, $m$, will remain fairly small. This is considered reasonable since large word lengths create a vast number of transition probabilities to be estimated, and the increase in dimension does not have much effect on the accuracy of the estimates. Therefore, we make the choice to limit word lengths to not be greater than $10$.  This is a pragmatic limit, and could be increased if required. Hence,  the word length that best represents the data is given by the corresponding model with the lowest AIC value.

\begin{theorem}[Akaike Information Criterion (AIC),  $m \ge 1$] \label{AICC}
\rm{Given a binary sequence $x$, the AIC for a de Bruijn process model with word length $m$ is given as:
\begin{equation}
AIC_{\text{P}} = 2^{m+1} - 2 \cdot \text{log} \big(\mathcal{L}_{P}(p|x) \big),
\end{equation}
where the profile likelihood is defined as:
\begin{equation}
\mathcal{L}_{\text{P}}(p|x) = \sum_{i=0}^{2^{m+1}-1} \frac{1}{n_{\frac{1}{2} (i - (i \hspace{0.1cm} \text{mod} 2))}^{i \hspace{0.1cm} \text{mod} 2^{m}} +1},
\end{equation}
where $n_{\frac{1}{2} (i - (i \hspace{0.1cm} \text{mod} 2))}^{i \hspace{0.1cm} \text{mod} 2^{m}}$ is the number of times that the transition $p_{\frac{1}{2} (i - (i \hspace{0.1cm} \text{mod} 2))}^{i \hspace{0.1cm} \text{mod} 2^{m}}$ occurs in the sequence $x$.}
\end{theorem}

\begin{proof}
See Appendix
\end{proof}

\subsection{Word Length Estimation: Bayes' Factors}

A Bayesian alternative to AIC is to proceed by comparing models using Bayes' factors \citep{Kass1995, OHagan1997}.  The method of Bayes' factors considers whether an observation $x$ of binary random variables was generated from either a word length $m_{1}$ de Bruijn process (hypothesis 1) with probability $P(x | m_{1})$, or from a length $m_{2}$ de Bruijn process (hypothesis 2) with probability $P(x | m_{2})$.  The prior probabilities $P(m_{1})$ and $P(m_{2})$ are also defined,  giving the probability that the sequence was indeed generated using a length $m_{1}$ or $m_{2}$ de Bruijn process respectively.  When combined with the data, this then gives appropriate posterior probabilities $P(m_{1}| x)$ and $P(m_{2}| x) = 1 - P(m_{1}| x)$. If Bayes' theorem is considered in terms of an odds scale of these hypotheses when in favour of $m_{1}$,  we have the following relationship:
\begin{equation}
\frac{P(m_{1} | x)}{P(m_{2} | x)} = \frac{P(x | m_{1})}{P(x | m_{2})} \frac{P(m_{1})}{P(m_{2})}.
\end{equation} 
If we say that the hypotheses $m_{1}$ and $m_{2}$ are equally likely then the Bayes' factor can be defined to be the posterior odds in favour of $m_{1}$:
\begin{equation}
B_{1,2} = \frac{P(x | m_{1})}{P(x | m_{2})}.
\end{equation}

Since the transition probabilities are unknown parameters in this case,  an expression for $P(x | m_{k})$ can be found by integrating over the parameter space. This becomes:
\begin{equation}
P(x | m_{k}) = \int P(x | p_{k}, m_{k}) P(p_{k} | m_{k}) \hspace{0.2cm} dp_{k},
\end{equation}
for $k \in \{1,2\}$, where $P(x | p_{k}, m_{k}) = \mathcal{L}(p_{k}, m_{k} | x)$ is the likelihood of the data and $P(p_{k} | m_{k})$ is the prior density of the model parameters, $p$. There is an obvious similarity between this expression and the model evidence in Theorem \ref{MEM}. Due to the conjugate priors,  the Bayes' factor ratio is stated to be the ratio of the model evidences for each of the model hypotheses, which is shown in Theorem \ref{MEM2}. We note here that it is not necessary to calculate the posterior on the transition probabilities since the expression for the Bayes' factor is only dependent on the prior density. For each calculation of $P(x | m_{k})$, we will know the length $m_{k}$ and hence the quantity of parameters, $p$, which are to be estimated. 

As with the AIC approach, we choose to limit the potential word lengths to not exceed $m=10$. Hence, to choose the word length that best represents the data,  calculate $B_{i,j} = \frac{P(x | m_{i})}{P(x | m_{j})}$ for each pair of models where $i,j \in \{1,2,...,10\}$ and select the value for $m$ in which the Bayes' factor is consistently higher. When values of $B_{i,j}$ are large, this gives more evidence to reject the model with word length $m_{i}$ in favour of the model with word length $m_{j}$. 

By only selecting 10 different models to compare and choosing the one that best represents the data, we acknowledge that we have implemented a frequentist aspect to our method. Instead, we could opt to do this in a fully Bayesian way to maximise the Bayes' factor and allow any word length to be considered. However, to do this we would have to put a fairly strong prior on $m$ to minimise large potential word lengths. This is left for future work.

\begin{theorem}[Estimation of De Bruijn Word length by Bayes' factors, $m \ge 1$] \label{MEM2}
\
\rm{Consider a binary observation $x=\{x_{1},  \ldots, x_{n}\}$ under one of two hypotheses. The first is a de Bruijn process with word length $m_{1}$ and the second is a de Bruijn process with word length $m_{2}$. The Bayes' factor ratio is as follows:}
\begin{equation}
B_{1,2} = \frac{P(x | m_{1})}{P(x | m_{2})}
\end{equation}
\rm{where,}
\begin{equation}
\begin{split}
P(x | m_{k}) & = \int P(x|p,m_{k})P(p|m_{k}) dp \\
& = \prod_{i=0}^{2^{m_{k}}-1} \frac{\Gamma(n_{i}^{((2i+1) \hspace{0.1cm} \text{mod } 2^{m_{k}}) - 1} + \beta_{i+1})\Gamma(n_{i}^{((2i+1) \hspace{0.1cm} \text{mod } 2^{m_{k}})}) + \alpha_{i+1})}{\Gamma(n_{i}^{((2i+1) \hspace{0.1cm} \text{mod } 2^{m_{k}}) - 1} + n_{i}^{((2i+1) \hspace{0.1cm} \text{mod } 2^{m_{k}})} +\beta_{i+1}+\alpha_{i+1})} ,
\end{split}
\end{equation}
\rm{for} $k \in \{1,2\}$.
\rm{When values of} $B_{1,2}$ \rm{are large, we have more evidence to reject the first hypothesis with word length} $m_{1}$ \rm{in favour of the second hypothesis with word length} $m_{2}$. 
\end{theorem}

\begin{proof}
Follows from Theorem \ref{MEM}
\end{proof}

\section{Examples II} \label{Examples2}

The following examples illustrate the method of inference so that given an observation $x$,  a de Bruijn process can be estimated. This includes estimating both the word length $m$, and the transition probabilities, $p_{i}^{j}$. 

The first sequence is an anti-clustered alternating sequence given by the top panel in Figure \ref{Samp2I}. This sequence is of length $n=200$ and was generated using an $m=2$ de Bruijn process with transition probabilities: $\{ p_{00}^{01}, p_{01}^{11}, p_{10}^{01}, p_{11}^{11} \} = \{ 0.9, 0.25, 0.75, 0.1\}$. We begin by trying to estimate the word length using the method of Bayes' factors with the model evidence stated in Theorem \ref{MEM2}. The model evidence for the sequence is calculated for each possible de Bruijn process with word lengths, $m = 1, ..., 10$, for comparison.  We do not assume any prior knowledge, so let each $\alpha = \beta = 1$ for the equivalence of a uniform prior. After calculating the Bayes' factor ratio for each pair of proposed models, we can conclude that the sequence was most likely generated with a length $m=2$ de Bruijn process.

\begin{figure}[ht]
\centering
\includegraphics[scale=0.5]{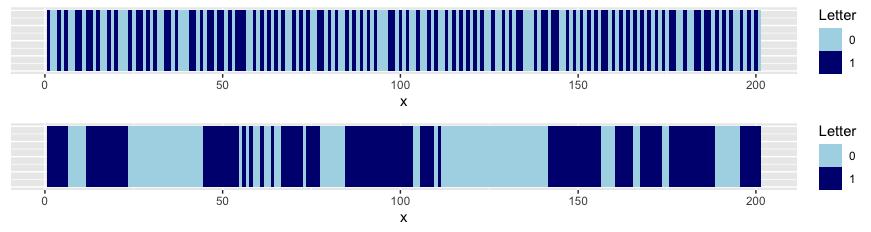}
\caption{Samples from length $m=2$ (top) and length $m=3$ (bottom) de Bruijn processes with letters 0 and 1. The transition probabilities are: $\{ p_{00}^{01}, p_{01}^{11}, p_{10}^{01}, p_{11}^{11} \} = \{ 0.9, 0.25, 0.75, 0.1 \}$ and $\{p_{000}^{001}, p_{001}^{011}, p_{010}^{101}, p_{011}^{111}, p_{100}^{001}, p_{101}^{011}, p_{110}^{101}, p_{111}^{111} \} = \{0.1, 0.7, 0.5, 0.8, 0.2, 0.5, 0.3, 0.9 \}$ respectively.}
\label{Samp2I}
\end{figure}

The left plot in Figure \ref{Hist2I} shows a histogram of the estimated word lengths for the $m=2$ de Bruijn process. A total of $1000$ sequences were generated from the de Bruijn process represented by top sequence in Figure \ref{Samp2I}, and the word length was estimated for each given sequence. The histogram shows that nearly every sequence was estimated to be generated from a length $m=2$ de Bruijn process. 

\begin{figure}[ht]
\centering
\includegraphics[scale=0.45]{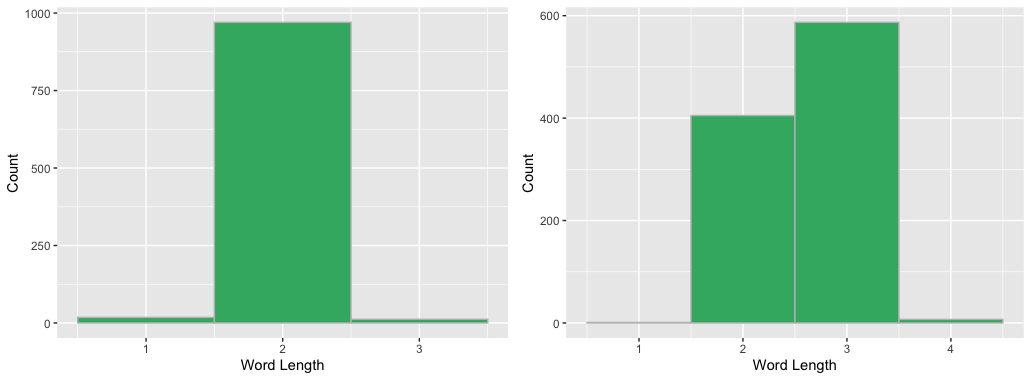}
\caption{Histograms of estimated word lengths from 1000 sequences generated from the $m=2$ (left) and $m=3$ (right) examples in Figure \ref{Samp2I}.}
\label{Hist2I}
\end{figure}

Finally, given the word length was estimated to be $m=2$, we then estimated the given transition probabilities using a simple Metropolis Hastings MCMC approach. Using the likelihood of the sequence given in Lemma \ref{TL2} and non-informative priors, the estimated parameters are given in Table \ref{TEInf} (left) along with $95\%$ confidence intervals. All estimates are close to the true values, where the true values lie within the confidence intervals for all parameters. On average, the true values and the expected values differ by $0.018$ indicating a small level of uncertainty.

The estimates of the parameters improve greatly with the increased length of the sequence. Table \ref{TEInf2} shows the effects of altering the length of the sequence. Using the same de Bruijn process with transition probabilities $\{p_{00}^{00}, p_{01}^{11}, p_{10}^{01}, p_{11}^{11}\} = \{0.9, 0.25, 0.75, 0.1\}$,  a total of $100$ sequences were each generated of lengths $n=50$, $n=100$, $n=200$ and $n=500$. The transition probabilities were then estimated for each of these sequences and Table \ref{TEInf2} gives the average estimate along with a $95\%$ interval. Although the average estimate does not change much, it is clear that the interval of possible values reduces as the length of the sequence increases.

\begin{table}[h!]
\centering
\begin{tabular}{||c | c c ||} 
\hline
$p_{i}^{j}$ & True value & Estimate [$95\%$ CI] \\ [0.5ex]
\hline\hline
$p_{00}^{01}$ & $0.9$ & $0.932 \hspace{0.3cm} [0.883,0.969]$ \\ 
$p_{01}^{11}$ & $0.25$ & $0.245 \hspace{0.3cm} [0.206,0.286]$ \\
$p_{10}^{01}$ & $0.75$ & $0.737 \hspace{0.3cm} [0.696,0.775]$ \\
$p_{11}^{11}$ & $0.1$ & $0.079 \hspace{0.3cm} [0.024,0.136]$ \\
& & \\
& & \\
& & \\
& & \\ [1ex] 
\hline
\end{tabular}
\quad
\begin{tabular}{||c | c c ||} 
\hline
$p_{i}^{j}$ & True value & Estimate [$95\%$ CI] \\ [0.5ex]
\hline\hline
$p_{000}^{001}$ & $0.1$ & $0.132 \hspace{0.3cm} [0.091,0.166]$ \\ 
$p_{001}^{011}$ & $0.7$ & $0.691 \hspace{0.3cm} [0.582,0.793]$ \\
$p_{010}^{101}$ & $0.5$ & $0.516 \hspace{0.3cm} [0.383,0.650]$ \\
$p_{011}^{111}$ & $0.8$ & $0.737 \hspace{0.3cm} [0.683,0.844]$ \\
$p_{100}^{001}$ & $0.2$ & $0.228 \hspace{0.3cm} [0.139,0.323]$ \\ 
$p_{101}^{011}$ & $0.5$ & $0.504 \hspace{0.3cm} [0.382,0.589]$ \\
$p_{110}^{101}$ & $0.3$ & $0.393 \hspace{0.3cm} [0.276,0.499]$ \\
$p_{111}^{111}$ & $0.9$ & $0.894 \hspace{0.3cm} [0.868,0.918]$ \\ [1ex] 
\hline
\end{tabular}
\caption{Tables to show the estimates of the transition probabilities for the sequences in Figure \ref{Samp2I} with transition probabilities $\{ p_{00}^{01}, p_{01}^{11}, p_{10}^{01}, p_{11}^{11} \} = \{ 0.9, 0.25, 0.75, 0.1 \}$ (left) and $\{p_{000}^{001}, p_{001}^{011}, p_{010}^{101}, p_{011}^{111}, p_{100}^{001}, p_{101}^{011}, p_{110}^{101}, p_{111}^{111} \} = \{0.1, 0.7, 0.5, 0.8, 0.2, 0.5, 0.3, 0.9 \}$ (right). The true value is given along with the estimate and $95\%$ confidence interval.}
\label{TEInf}
\end{table}

\begin{table}[h!]
\centering
\begin{tabular}{||c | c c c c ||} 
\hline
$p_{i}^{j}$ & $n=50$ & $n=100$ & $n=200$ & $n=500$ \\ [0.5ex] 
\hline\hline
$p_{00}^{01}=0.9$ & $0.905 \hspace{0.3cm} [0.687,0.999]$ & $0.902 \hspace{0.3cm} [0.768,0.999]$ & $0.902 \hspace{0.3cm} [0.827,0.998]$ & $0.901 \hspace{0.3cm} [0.854,0.954]$ \\ 
$p_{01}^{11}=0.25$ & $0.262 \hspace{0.3cm} [0.098,0.424]$ & $0.249 \hspace{0.3cm} [0.149,0.364]$ & $0.251 \hspace{0.3cm} [0.176,0.315]$ & $0.249 \hspace{0.3cm} [0.203,0.287]$ \\
$p_{10}^{01}=0.75$ & $0.745 \hspace{0.3cm} [0.553,0.905]$ & $0.753 \hspace{0.3cm} [0.641,0.848]$ & $0.757 \hspace{0.3cm} [0.697,0.813]$ & $0.754 \hspace{0.3cm} [0.704,0.787]$ \\
$p_{11}^{11}=0.1$ & $0.082 \hspace{0.3cm} [0.001,0.287]$ & $0.082 \hspace{0.3cm} [0.001,0.222]$ & $0.080 \hspace{0.3cm} [0.001,0.178]$ & $0.102 \hspace{0.3cm} [0.057,0.156]$ \\ [1ex] 
\hline
\end{tabular}
\caption{Table showing the effects of altering the lengths of sequences on estimating transition probabilities. For each length ($n=50$, $n=100$, $n=200$, $n=500$), 100 sequences are generated from the $m=2$ de Bruijn process with transition probabilities $\{ p_{00}^{01}, p_{01}^{11}, p_{10}^{01}, p_{11}^{11} \} = \{ 0.9, 0.25, 0.75, 0.1 \}$ and the parameters are estimated. The average estimates along with $95\%$ intervals are given.}
\label{TEInf2}
\end{table}

The second example (bottom panel in Figure \ref{Samp2I}) is a length $n=200$ sequence designed such that there are large clustering blocks of $0$'s and $1$'s, with independent Bernoulli patterns occurring occasionally. It is generated using a length $m=3$ de Bruijn process with transition probabilities: $\{ p_{000}^{001}, p_{001}^{011}, p_{010}^{101}, p_{011}^{111}, p_{100}^{001}, p_{101}^{011}, p_{110}^{101}, p_{111}^{111} \} = \{ 0.1, 0.7, 0.5, 0.8, 0.2, 0.5, 0.3, 0.9 \}$.

As for the previous example, we begin by estimating the de Bruijn word most likely used to generate the sequence using Bayes' factors. This was estimated to be $m=3$. Although estimated correctly, by observing the right histogram in Figure \ref{Hist2I}, we can see that this isn't always the case for sequences generated from this de Bruijn process. After simulating $1000$ sequences and estimating the word length for each, almost $60\%$ are estimated to be $m=3$, but just over $40\%$ are actually estimated to be $m=2$. This is because some of the sequences generated have similar correlation structures to sequences generated with an $m=2$ de Bruijn process. The extra transition probabilities from the $m=3$ de Bruijn process have little effect and it is likely that these sequences could have been generated with transition probabilities close to $\{ p_{00}^{01}, p_{01}^{11}, p_{10}^{01}, p_{11}^{11} \} = \{ 0.1, 0.65, 0.35, 0.9\}$.

Finally, the transition probabilities (given $m=3$) are estimated and given in Table \ref{TEInf} (right) along with $95\%$ confidence intervals. All estimates are reasonably accurate, where on average, the estimates and true value differ by $0.020$. All true values lie within the $95\%$ confidence intervals.

\section{Applications} \label{Expp}

\subsection{Application I: Precipitation}

As an example application we investigate the distribution of run lengths of the presence (or absence) of observed precipitation and whether this can be successfully described using the de Bruijn process run length structure.  In order to do this, we have obtained daily weather recordings from a human-facilitated observation station used in the Global Historical Climatology Network-Daily (GHCN) database.  There are 455 daily measurements of precipitation recorded at a station in Eskdalemuir, UK from January 2021 to March 2022.  The data gives the total amount of precipitation (rain, melted snow, etc.) recorded in inches in a 24 hour period ending at the observation time. We further translate this to binary data, recording if either precipitation was observed ($1$) or no precipitation was observed ($0$). The full data is given in Figure \ref{Rain}.

\begin{figure}[ht]
\centering
\includegraphics[scale=0.35]{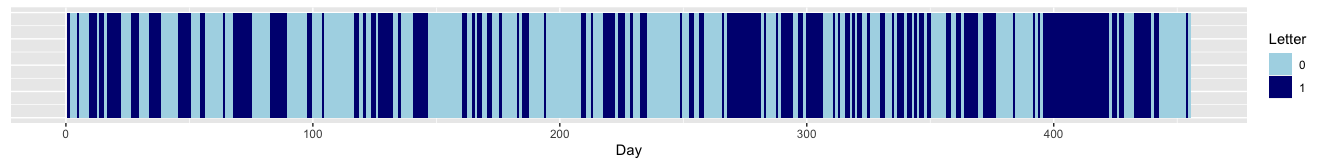}
\caption{Daily precipitation data recorded at a station in Eskdalemuir, UK. Dark blue (or $1$) represents that there was precipitation recorded in the 24 hour period, whilst light blue (or $0$) corresponds to no precipitation recorded that day.}
\label{Rain}
\end{figure}

Since the amount of precipitation can depend on the current season, the data is first split up into corresponding seasons (spring, summer, autumn and winter).  The corresponding distributions of run lengths are presented in the histograms in Figure \ref{RLhist2} in the above order of seasons. There are shown to be different distributions of run lengths for each season such that for the colder seasons, there are both many single days of rain as well as long periods of rain lasting up to 27 days. 

\begin{figure}[ht]
\centering
\includegraphics[scale=0.7]{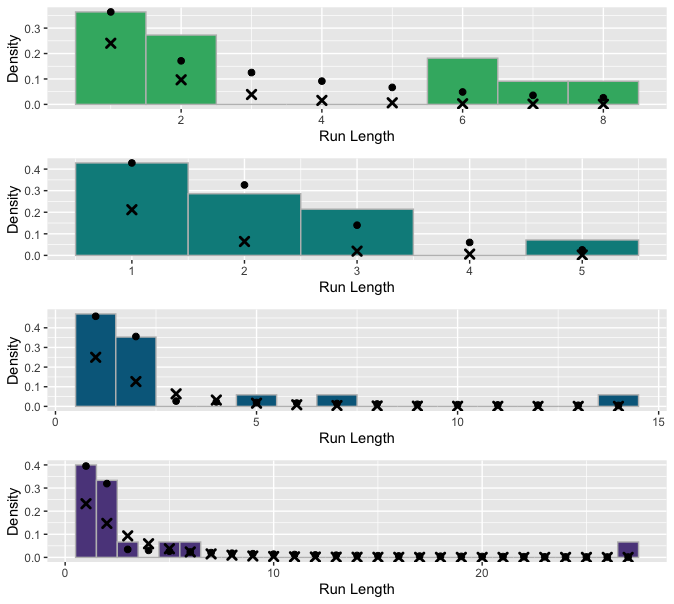}
\caption{Histograms showing the seasonal distributions of run lengths of recorded precipitation in a 24 hour period at a station in Eskdalemuir.  From top to bottom, the data corresponds to seasons spring, summer, autumn and winter from January 2021 to March 2022. Black dots give the theoretical run length distributions derived from estimated de Bruijn process. Black crosses give a geometric density fitted to the data.}
\label{RLhist2}
\end{figure}

The auto-correlation plots for each of the season time series are presented in Figure \ref{ACFRa}. It is clear from these plots that the time series for the spring precipitation has a de Bruijn structure, whilst there is little correlation between time points from the summer plot. We are likely to see large run lengths in the spring data, but more random patterns in the summer data.  We have similar results in the autumn plot as compared with the summer plot, but an increased correlation again between time points in the winter plot.

\begin{figure}[ht]
\centering
\includegraphics[scale=0.7]{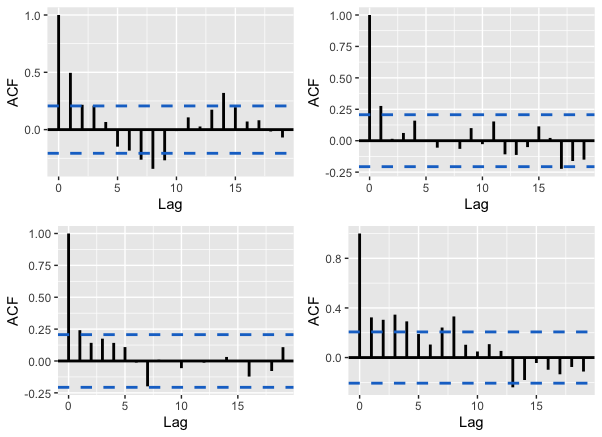}
\caption{Auto-correlation plots for the daily precipitation data from Figure \ref{Rain}. Each plot corresponds to the data from one of four seasons, top left: spring, top right: summer, bottom left: autumn, bottom right: winter.}
\label{ACFRa}
\end{figure}

A de Bruijn process is then fitted separately to each season, estimating both the word length and associated transition probabilities.  Both spring and summer were estimated to have a word length of $m=2$, whilst autumn and winter were instead estimated to have a word length of $m=3$.  The black dots seen in Figure \ref{RLhist2} correspond to the theoretical run lengths generated by applying the estimated transition probabilities to Lemma \ref{RLD2} and Theorem \ref{RLD3}. In general, the theoretical run lengths match well with the run lengths from the data implying that the data can indeed be represented by a de Bruijn structure. 

We have further compared the theoretical run lengths from the de Bruijn processes with that of a geometric distribution.  The black crosses in Figure \ref{RLhist2} give the probabilities of a run of length $n$ generated from a geometric distribution for each given season.  It is evident that in each season, the run lengths produced match those from the de Bruijn process far better than those produced from the geometric distribution. We can conclude that the precipitation recorded has correlation structures that are far more complex than that of Bernoulli trials or a Markov property.

\subsection{Application II: The Boat Race} \label{App}

As a further example, we apply the de Bruijn process inference to the annual boat race between Cambridge university and Oxford university. The race was first held in 1829 and has continued every year until present, with a few missing years (partly due to the first and second World Wars and COVID-19). There are 166 data points in total where Oxford has won the race 80 times, Cambridge has won 85 times, and they have drawn once (1877). We also note that two races occurred in 1849, and so for simplicity reasons we chose to ignore the second recorded race (where Oxford won) and the race where the two universities drew.

To model the data using the de Bruijn process, let $0$ represent the years that Oxford won the race and let $1$ represent the years that Cambridge won the race. Hence $\pi(\rm{Oxford}) = \pi(0) = 79/164 = 0.482$ and $\pi(\rm{Cambridge}) = \pi(1) = 85/164 = 0.518$. The data is shown in Figure \ref{boat1} (where dark blue represents Oxford and light blue represents Cambridge).

\begin{figure}[ht]
\centering
\includegraphics[scale=0.25]{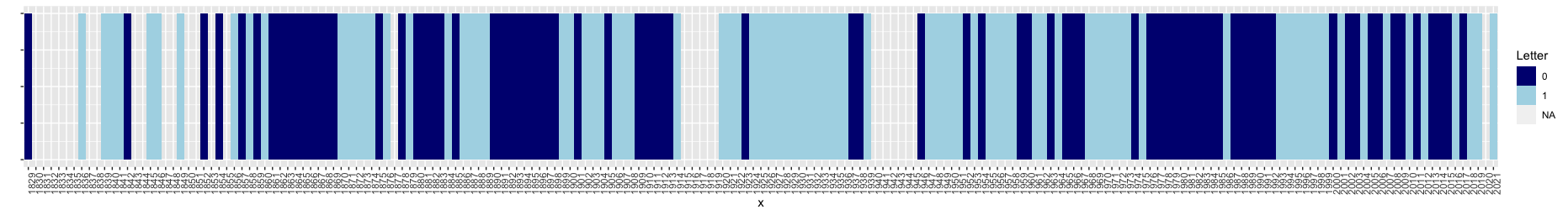}
\caption{Data for the annual boat race between Cambridge and Oxford universities. Races where Oxford have won are represented by a $0$ and shown in dark blue. Races where Cambridge have won are represented by a $1$ and are shown in light blue.}
\label{boat1}
\end{figure}

As before, we begin by estimating the word length for the estimated de Bruijn process. This is done using the Bayes' factor method explained in Section \ref{Inf} and is best estimated to be $m=2$.  Since the Bayes' factor method compares each pair of models in turn, we noticed that there was little difference between the $m=2$ and $m=3$ models, hence $m=3$ would still likely give valid results. Given both of these possible word lengths, the associated transition probabilities are estimated using Bayesian methods. These are given in Table \ref{btable2} along with $95\%$ credible intervals. The auto-correlation plot for the boat race time series is given in Figure \ref{ACFBo}. This also implies a de Bruijn structure is present with a small word length, suggesting again that either word length $m=2$ or $m=3$ is plausible.

\begin{figure}[ht]
\centering
\includegraphics[scale=0.6]{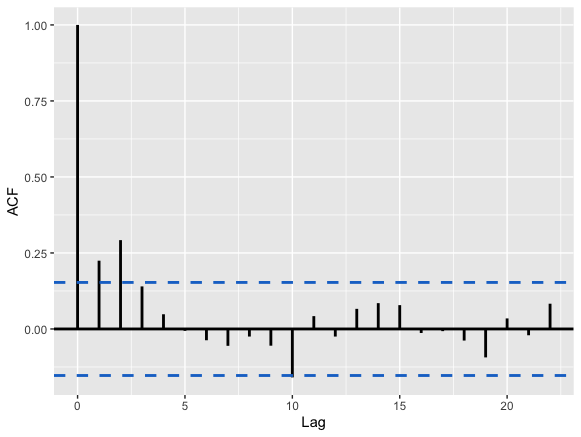}
\caption{Auto-correlation plot for the Oxford and Cambridge boat race data from Figure \ref{boat1}.}
\label{ACFBo}
\end{figure}

\begin{table}[h!]
\centering
\begin{tabular}{||c | c c ||} 
\hline
$p_{i}^{j}$ & Estimate & $95\%$ CI \\ [0.5ex]
\hline\hline
$p_{00}^{01}$ & $0.283$ & $[0.175, 0.386]$ \\ 
$p_{01}^{11}$ & $0.462$ & $[0.327, 0.601]$ \\
$p_{10}^{01}$ & $0.519$ & $[0.372, 0.652]$ \\
$p_{11}^{11}$ & $0.723$ & $[0.618, 0.813]$ \\
& & \\
& & \\
& & \\
& & \\ [1ex] 
\hline
\end{tabular}
\quad
\begin{tabular}{||c | c c ||} 
\hline
$p_{i}^{j}$ & Estimate & $95\%$ CI \\ [0.5ex]
\hline\hline
$p_{000}^{001}$ & $0.244$ & $[0.142, 0.364]$ \\ 
$p_{001}^{011}$ & $0.458$ & $[0.237, 0.686]$ \\
$p_{010}^{101}$ & $0.362$ & $[0.175, 0.551]$ \\
$p_{011}^{111}$ & $0.817$ & $[0.625, 0.955]$ \\
$p_{100}^{001}$ & $0.340$ & $[0.188, 0.594]$ \\ 
$p_{101}^{011}$ & $0.467$ & $[0.251, 0.691]$ \\
$p_{110}^{101}$ & $0.671$ & $[0.444, 0.833]$ \\
$p_{111}^{111}$ & $0.680$ & $[0.549, 0.789]$ \\ [1ex] 
\hline
\end{tabular}
\caption{Table to show the estimated transition probabilities for the boat race data given in Figure \ref{boat1} when the word length is given to be either $m=2$ (left) or $m=3$ (right). Estimates are given alongside $95\%$ credible intervals.}
\label{btable2}
\end{table}

Initially considering the case when $m=2$, the marginal probabilities for the letters and words are given as $\pi(0) = 0.482$, $\pi(1) = 0.518$ and $\{\pi(00), \pi(01), \pi(10), \pi(11)\} = \{0.303, 0.191, 0.184, 0.322\}$. So far Cambridge have won more races than Oxford and there were also more times that Cambridge won consecutively. From the transition probability estimates, there is a $46.2\%$ chance for Cambridge to win again if they only won the previous race. This increases to a $72.3\%$ chance for Cambridge to win again if they have already won the past two races. Equivalently for Oxford, there is a $48.1\%$ chance for them to win again if they only won the previous race. But this again increases to a $71.7\%$ chance for Oxford to win again if they have already won the past two races.

We can hence conclude that both universities increase their chances of winning the next race if they have won previous races consecutively. This is slightly higher for Cambridge (increase of $0.6\%$ compared to Oxford when considering the two previous races) and is likely to be due to having overlapping team members in these years or having more confidence to win again if they have done previously.

If the sequence of races is $01$, then there is a $53.8\%$ chance for Oxford to win the next race. Alternatively if the sequence is $10$, then there is a $51.9\%$ chance for Cambridge to win the next race. There is almost a $50\%$ chance for either university to win the next race if the opposite university won the last race. Oxford have a slightly higher chance of winning, implying that they may have a larger drive to win again if Cambridge won the last race.
 
If we let $m=3$, the marginal probabilities for the letters and words are instead estimated to be $\pi(0) = 0.482$, $\pi(1) = 0.518$ and $\{\pi(000), \pi(001), \pi(010), \pi(011), \pi(100), \pi(101), \pi(110), \pi(111)\} = \{0.226, 0.089, 0.096, 0.082, 0.089, 0.096, 0.089, 0.233\}$.

Given $m=3$, there is now a $81.7\%$ chance for Cambridge to win the next race if they have won the last two races. This falls to a $68.0\%$ chance of winning if they have won the last three races. Hence, there is a much higher chance for Cambridge to win consecutive races, but the increased chance in winning reduces when the number of wins in a row exceeds two. Again, this is likely to be due to the team members taking part, winning tactics being used or previous wins providing increased confidence. The individual rowers are fairly likely to take part in the race for two years in a row, but less likely to take part in further races due to finishing their studies at the university. 

Oxford now only has a $66.0\%$ chance to win the next race if they have won the last two races, which now increases to a $75.6\%$ chance to win the race if they have won the last three races. Unlike Cambridge, Oxford increases their chance of winning the more races they win in row, indicating that they tend to value past tactics or confidence from winning over individual team members.

We observe that as well as having a higher chance of a longer streak of winning, Oxford also tends to perform better when neither university has built up a run in previous races. When the sequences of races is $101$, the probability of Oxford winning is $53.3\%$, whilst when the sequence is $010$, Cambridge has a $36.2\%$ chance of winning the next race. This is similar when we just consider the previous race. If Cambridge won the previous race then they only have a $46.3\%$ chance of winning the next race (calculated using law of total probability), compared with if Oxford won the previous race then Oxford have a $47.8\%$ chance of winning again.

We can also use the de Bruijn process to predict the results of the 2022 boat race. If $m=2$, then there is a $59.7\%$ ($p_{01}^{11} \pi(0) + p_{11}^{11} \pi(1) = 0.597$) chance of Cambridge winning the race and a $40.3\%$ of Oxford winning the race in 2022. If $m=3$, then there is a $57.7\%$ ($p_{101}^{011} \pi(0) + p_{111}^{111} \pi(1) = 0.577$) chance of Cambridge winning the race and a $42.3\%$ chance of Oxford winning the race in 2022.  Oxford won the race in 2022.

\section{Discussion} \label{Conclu}

We have introduced the de Bruijn process, a novel method for modelling correlated binary random variables where the spread of correlation is controlled by distance. It is often the case that variables close in distance are more likely to be highly correlated then those further away. Hence, by altering the correlation we are able to consider both clustered and anti-clustered sequences of $0$'s and $1$'s. 

The process uses structures from de Bruijn graphs. These are directed graphs where the nodes of the graph are length $m$ sequences of $0$'s and $1$'s. If a Markov property is placed on these length $m$ sequences rather than the variables themselves, then we are able to control the amount of correlation included in the neighbourhood. Examples are presented throughout as well as details of the stationary distribution for the length $m$ sequences. 

We have also presented a run length distribution in order to determine how clustered a sequence is likely to be from a given de Bruijn process. The run length distribution gives the probability of a run of $1$'s of length $n$ bounded by a $0$ at each end. From this, we are then able to calculate the expected run length, variance of run length and generating functions.  Given a sequence of correlated binary data, we have also shown that we can fit a de Bruijn process by estimating both the de Bruijn word length and associated transition probabilities. 

De Bruijn graphs can quickly become very complicated with many possible transition probabilities for large word lengths. Therefore, we believe that it would be useful to try and limit the number of transition probabilities, which could be particularly important for inference. We also intend to develop non-stationary de Bruijn processes where we would be able to create chains that are clustered in some places and anti-clustered in others.

\bibliographystyle{abbrvnat}
\bibliography{DBGPFbib}
\end{document}